\documentclass[11pt, letterpaper, reqno]{amsart}

\usepackage[numbers,sort]{natbib}
\setcitestyle{square}
\usepackage{xpatch}
\usepackage{amsfonts}
\usepackage{amsmath}
\usepackage{amsthm}
\usepackage{amssymb}
\usepackage{dsfont}
\usepackage[normalem]{ulem}
\usepackage{bbm}
\usepackage{bm}
\usepackage{url}
\usepackage{mathrsfs} 
\usepackage{setspace,lipsum}
\usepackage[shortlabels]{enumitem}
\usepackage{subcaption}
\usepackage{graphicx}
\usepackage{xcolor}
\usepackage{hyperref}

\hypersetup{
    colorlinks = true,
    linkcolor={red!50!black},
    citecolor={red!50!black},
    urlcolor={blue!80!black}
}

\usepackage[nameinlink,capitalise]{cleveref}

\usepackage{thmtools, thm-restate}

\newtheorem{theorem}{Theorem}[section]
\newtheorem*{theorem*}{Theorem}

\newtheorem{lemma}[theorem]{Lemma}
\newtheorem{proposition}[theorem]{Proposition}
\theoremstyle{definition}
\newtheorem{definition}[theorem]{Definition}

\theoremstyle{remark}

\newtheorem{conjecture}[theorem]{Conjecture}

\allowdisplaybreaks
\usepackage{lineno}

\usepackage{tikz}
\usetikzlibrary{shapes.geometric, arrows}

\tikzstyle{block} = [rectangle, rounded corners, minimum width=3cm, minimum height=1cm, text centered, draw=black]
\tikzstyle{arrow} = [thick,->,>=stealth]

\Crefname{assumption}{assumption}{assumptions}

\newcommand{\overbar}[1]{\mkern 1.5mu\overline{\mkern-1.5mu#1\mkern-1.5mu}\mkern 1.5mu}

\numberwithin{equation}{section} 

\colorlet{yyred}{red!50!black}

\makeatletter
\xpatchcmd\NAT@citex
 {%
  \@citea\NAT@hyper@{%
    \NAT@nmfmt{\NAT@nm}%
    \hyper@natlinkbreak{\NAT@aysep\NAT@spacechar}{\@citeb\@extra@b@citeb}%
    \NAT@date
  }%
 }
 {%
  \@citea
  \NAT@nmfmt{\NAT@nm}%
  \NAT@aysep\NAT@spacechar
  \NAT@hyper@{\NAT@date}%
 }
 {}{}
\xpatchcmd\NAT@citex
 {%
  \@citea\NAT@hyper@{%
    \NAT@nmfmt{\NAT@nm}%
    \hyper@natlinkbreak{\NAT@spacechar\NAT@@open\if*#1*\else#1\NAT@spacechar\fi}%
    {\@citeb\@extra@b@citeb}%
    \NAT@date
  }%
 }
 {
  \@citea
    \NAT@nmfmt{\NAT@nm}%
    \NAT@spacechar\NAT@@open\if*#1*\else#1\NAT@spacechar\fi
    \NAT@hyper@{\NAT@date}%
 }
 {}{}
\makeatother

\newcommand{\E}{\mathbb{E}}

\title[]{Novel Lower Bounds on M/G/k Scheduling}

\author[Z. Wang]{Ziyuan Wang}
\thanks{Department of Industrial Engineering and Management Sciences, Northwestern University, Evanston, United States. \texttt{ziyuanwang2027@u.northwestern.edu}}

\author[I. Grosof]{Izzy Grosof}
\thanks{Department of Industrial Engineering and Management Sciences, Northwestern University, Evanston, United States.
\texttt{izzy.grosof@northwestern.edu}}

\date{\today}
\subjclass[2010]{}

\begin{document}

\begin{abstract}
In queueing systems, effective scheduling algorithms are essential for optimizing performance. Optimal scheduling for the M/G/k queue has been explored in the heavy traffic limit, but much remains unknown in the intermediate load regime.

In this paper, we give the first framework for proving nontrivial lower bounds on the mean response time of the M/G/k system under arbitrary scheduling policies. Our bounds tighten previous naive lower bounds by more than 60\%, yielding significant improvements particularly for moderate loads. Key to our approach is a new variable-speed queue,  which more accurately captures the work completion behavior of multiserver systems. To analyze the expected work of this queue, we develop a novel manner of employing the drift method or the BAR approach, by developing test functions via the solutions to a differential equation.

We validate our results numerically for systems with up to 5 servers and a range of job size distributions.
\end{abstract}

\maketitle

\section{Introduction}\label{sec:intro}

In queueing systems, effective scheduling algorithms are essential for optimizing performance in a wide range of modern applications. In the analysis of queueing systems, one of the most important performance metrics is the mean response time, denoted as $\E[T]$. An object of particular interest is the optimal scheduling policy, which minimizes the mean response time among all policies. In the single-server setting, optimal scheduling is well understood. 

However, many modern queueing systems are multiserver systems. In particular, many practical systems operate with a small number of parallel servers and preemptible workloads with known and predictable sizes. Examples include personal devices such as cell phones or lower-end laptops with 2–8 cores processing short machine learning tasks, a single GPU partitioned across multiple applications, file transfers over a network with a limited number of concurrent connections, and hospital emergency rooms where triage allows urgent cases to preempt routine consults. Thus, a deeper understanding of optimal scheduling in the multiserver setting with preemption and predictable job sizes is required.

All of the above examples involve preemptible jobs. For instance, training a small machine learning task incurs only minor overhead when pausing and storing model parameters and gradient information. File transfers and downloads can be paused and resumed when needed. Preemption in hospital emergency departments is even more straightforward, as urgent cases clearly take precedence over routine ones.

Moreover, in each of these cases, job duration is at least partially predictable in advance. Machine learning tasks often have relatively predictable durations (\citet{Hutter2014}), as a function of the number and size of predictors as well as the complexity of the model. GPU workloads consist largely of matrix multiplications with with accurately measurable operations per second (\citet{NvidiaGPUPerfGuide}). In most file transfer settings, file sizes are known. The triage in emergency departments provides some indication of how long a treatment is likely to take.

The seminal work of \citet{Schrage1968Proof} proves that the Shortest Remaining Processing Time policy (SRPT) minimizes the mean response time of jobs in an $M/G/1$ queue with preemption and known sizes. \citet{Schrage1966} also characterize the mean response time of SRPT exactly. The single-server optimality of SRPT has been used to prove optimality results in the $M/G/k$ under asymptotic load conditions. When the job sizes are known upon arrival, \citet{Grosof2018} show that the multiserver SRPT  is optimal in the heavy-traffic regime, as the load $\rho$ approaches the capacity of the system. Likewise, when job sizes are unknown, single-server results have been used to prove multiserver optimality.
If the job sizes are unknown but drawn from a known distribution, the optimal scheduling policy in the $M/G/1$ is known to be the Gittins index policy (\citet{Gittins1979}).
Correspondingly, \citet{Scully2020} show that the multiserver Gittins policy is optimal in the heavy-traffic regime. 
At the opposite extreme, in the light traffic limit, it is rare for nontrivial scheduling options to be available, so the choice of scheduling policy is less important.

However, much less is known about optimal $M/G/k$ scheduling under moderate loads. This intermediate regime often reflects the conditions faced by many real-world systems. Recently, the SEK policy introduced by \citet{Grosof2026} is proven to beat SRPT-$k$ in the moderate load regime. The fact that SRPT-$k$ is not optimal raises an important question, 

\begin{quote}
    \textit{How much improvement beyond SRPT-$k$ and SEK is possible under moderate load?}
\end{quote}

Specifically, we aim to prove tighter lower bounds on the mean response time of $M/G/k$ than the existing naive lower bounds systems under arbitrary scheduling policies, see \Cref{sec:main}. We focus on the case where job sizes are known in advance, and discuss how to extend our framework to handle unknown or estimated sizes in \Cref{sec:generalization}.

\subsection{Challenges in analyzing $M/G/k$ scheduling}
\label{sec:challenges}

The analysis of arbitrary scheduling policies in $M/G/k$ systems presents substantial challenges compared to the single-server case, and there have been limited results regarding the mean response time of $M/G/k$ systems under arbitrary scheduling policies.

The \textit{tagged-job} method (\citet{Harchol-Balter2013}) is a classical tool for analyzing single-server queues under many different scheduling policies, including the broad SOAP class of policies (\citet{Scully2018}). However, the tagged-job method breaks down in the analysis of multiserver scheduling. Unlike the single-server setting, the rate at which a multiserver queue completes work varies, making it intractable to quantify the random amount of work encountered by the tagged job in the system.

As such, many results focus on the heavy-traffic limit, where the work completion behavior of an $M/G/k$ system is nearly identical to that of a single-server queue, allowing the tagged-job method to be reintroduced. For example, \citet{Grosof2018} prove that multiserver SRPT is heavy-traffic optimal using a multiserver tagged-job method. Unfortunately, such results are not tight when applied to scheduling in the intermediate load regime. 

Another more flexible tool, the drift method (\citet{Eryilmaz2012}), likewise encounters additional challenges when applied to queues with variable work completion rate, and has not been employed in multiserver systems with arbitrary scheduling.

In addition, there have been approximation results for $M/G/k$ systems under specific policies, such as the First-Come First-Serve (FCFS) policy (\citet{Gupta2011} and \citet{Gupta2010}). There have also been attempts at exact solutions for multiserver queues using matrix analytic methods, but the state of the art is limited to FCFS scheduling and scheduling policies with no more than two priority classes (\citet{Sleptchenko2005}).

\subsection{Our contributions}

In this paper, we present the first nontrivial lower bounds on the mean response time of any $M/G/k$ system.

In the prior literature, only two lower bounds on mean response time for the $M/G/k$ queue under arbitrary scheduling policies have appeared. These bounds are the mean service time and the \emph{resource-pooled} single-server SRPT response time (see \Cref{sec:wine}). 

In this paper, we introduce three new lower bounds that significantly tighten these prior bounds: the \emph{MixEx bound}, the \emph{ISQ bound}, and the \emph{ISQ-Recycling bound}, see \Cref{sec:aux_isq} for our terminology regarding recycling. These bounds progressively improve upon one another and their degree of improvement are numerically explored in detail in \Cref{sec:empirical}. Our bounds hold for any arbitrary scheduling policy and hold across all system loads.

To derive these bounds, we introduce a novel single-server variable-speed queue, the Increasing Speed Queue (ISQ-$k$). ISQ-$k$ allows us to bridge the multiserver and single-server queues in this challenging intermediate-load environment. Moreover, we develop a novel DiffeDrift extension to the drift method to analyze variable speed queues, including ISQ-$k$.

To lower bound mean response time, we start by proving that the mean total work of the ISQ-$k$ lower bounds the mean total work of any $M/G/k$ system under arbitrary scheduling policies and at all loads (see \Cref{sec:isq}).
We build on this result to lower bound the mean \emph{relevant} work of $M/G/k$ systems under arbitrary scheduling policies.

Our analysis of the ISQ-$k$ system is based on the drift method, which relies on the careful selection of test functions (\citet{Eryilmaz2012}). However, most prior applications of drift methods either use standard test functions that are not customized to the system or employ rescaling techniques to eliminate variable rates (\citet{Braverman2017} and \citet{Hurtado-Lange2020}). Neither of these methods can accommodate complex variable work completion rates and give a tight analysis of the ISQ-$k$ system.

To overcome this limitation, we introduce a novel method of deriving test functions, which we term the DiffeDrift method. This consist of formulating our test functions as solutions to differential equations.

We discuss the DiffeDrift method in \Cref{sec:derive}. The test functions derived in this manner are not covered by existing approaches to the drift method, specifically the existing Basic Adjoint Relationship (BAR) results (see \Cref{sub_sec:prior_bar}). The standard approach would involve a sequence of truncations of these test functions, which would be highly complex given their nontrivial form. Instead, we prove a novel BAR result that accommodates our test functions in \Cref{sec:drift}.

Our DiffeDrift method allows us to select the drift first, then come up with the test function. We use this flexibility to handle the variable work completion rate of the ISQ-$k$ system.
We believe the DiffeDrift method and our BAR extension, \Cref{prop:drift}, have broad applicability to queueing systems.

\setlength{\belowcaptionskip}{0pt} 

\begin{figure}[htbp!]
    \centering
    \includegraphics[width=0.85\textwidth]{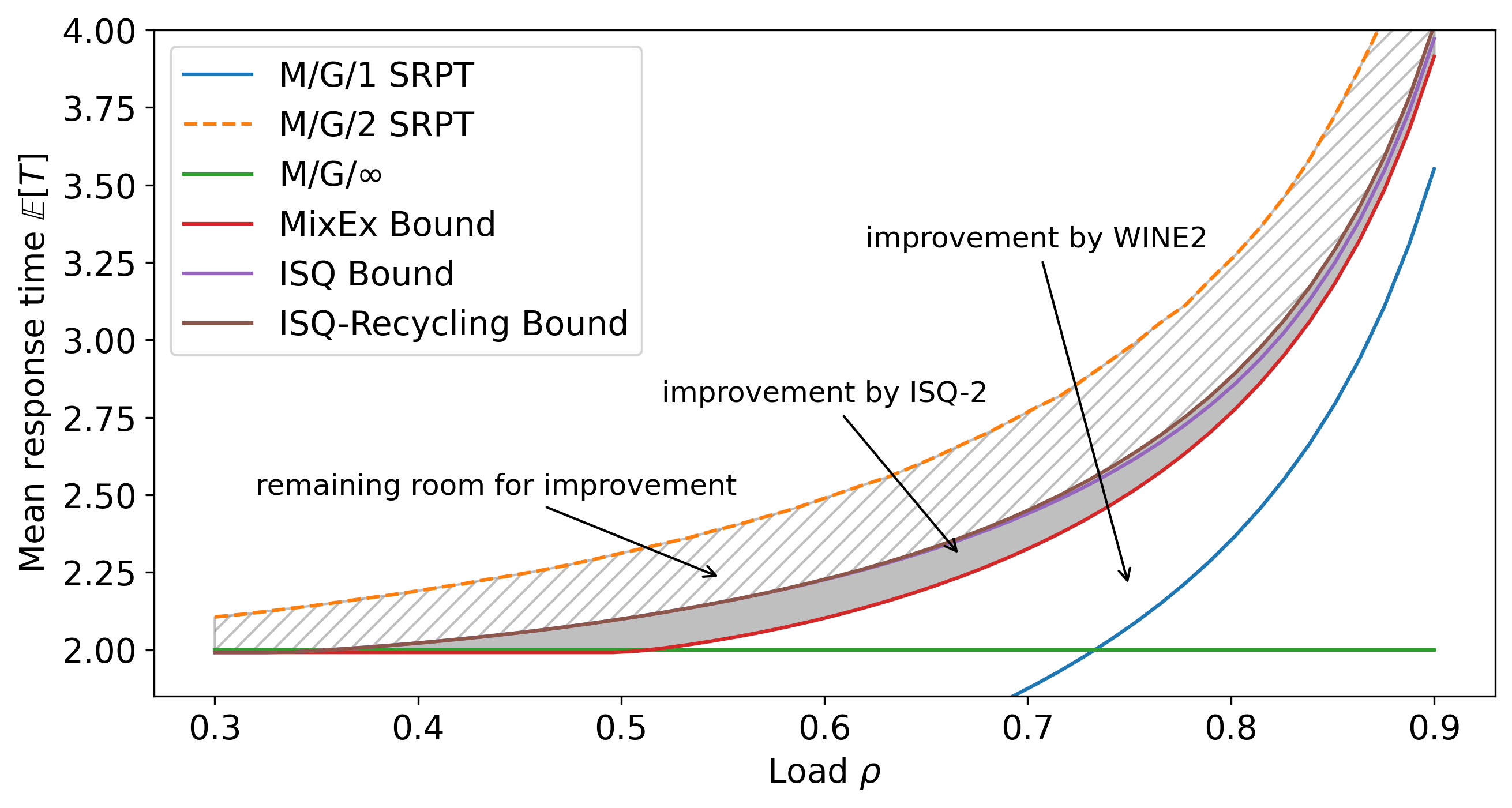}
    \caption{$M/G/2$ mean response time under SRPT-2, compared against naive lower bounds ($M/G/1$/SRPT and $M/G/\infty$, and against our lower bounds. See \Cref{sec:wine} for the definition of MixEX, ISQ and ISQ-Recycling bounds.}
    
    \label{fig:intro_wine}
\end{figure}

\subsection{Improvement upon prior results}
We now show the massive improvement of our lower bounds over the prior naive lower bounds, all of these bounds. Consider a two-server setting with an exponential job size distribution, as shown in \Cref{fig:intro_wine}. The blue line represents the mean response time of the resource pooled SRPT system, and the green line represents the mean service time bound. The empirical mean response time of $M/G/2$-SRPT is given by the dashed orange curve. While SRPT-2 is known to not be optimal in this setting (\citet{Grosof2026}), it gives a useful sense of where the optimal policy may lie. Both lower bounds are observed to be loose in the intermediate load regime. 

The \textit{MixEx bound} given by the red curve, represents our novel combination of the mean service time bound and the resource-pooled SRPT bound via the WINE method (see \Cref{sec:wine}). This bound significantly improves upon the naive bounds, particularly in the intermediate regime. The purple curve is the \textit{ISQ bound}, which incorporates our novel ISQ-$k$ system. Finally, we call our strongest lower bound, represented by the brown curve in \Cref{fig:intro_wine}, the the \textit{ISQ-Recycling bound}.

To quantify the degree of improvement of our novel lower bounds, we introduce a key performance metric, the \textit{Uncertainty Improvement Ratio} (UIR), which we define in \Cref{def:uir}. UIR measures the fraction of the uncertainty region closed by a given bound relative to the best prior lower bound. More concretely, the uncertainty region is defined as the gap between the simulated SRPT-$k$ response time and the strongest available prior bound. A new bound achieves a UIR equal to the percentage of this gap that it eliminates.

With the help of ISQ-$k$, our ISQ-Recycling bound provides a substantial improvement over our MixEx, which itself represents a significant advance over the naive bounds. In \Cref{fig:intro_UIR} with exponential job size distribution, MixEx achieves a maximum UIR of about 50\% relative to the naive bounds. Building on this, the ISQ-Recycling bound attains a UIR exceeding 30\% relative to MixEx in the same setting and for a total UIR of 62\% relative to the naive bound, while maintaining meaningful improvement across the entire medium-load range ($\rho = 0.4$ to $0.95$). Moreover, in the low variability setting, the ISQ-Recycling bound attains a notably high UIR of 70\% relative to MixEx bound (see \Cref{sec:varying_jobsizes}). 

By contrast, the recently introduced SEK policy (\citet{Grosof2026}), which was proven to outperform SRPT-$k$ at all load levels and aims to tighten upper bounds through improved simulated response times, achieves only a single-digit empirical UIR of 3\% (see \Cref{apdx:sek}). This highlights the scale of our improvement: while prior advances in upper bounds have yielded only modest gains, our framework achieves order-of-magnitude more UIR than new upper bounds, substantially narrowing the uncertainty region (see \Cref{sec:upper_bound}).

\begin{figure}[htbp!]
    \centering
    \includegraphics[width=0.85\textwidth]{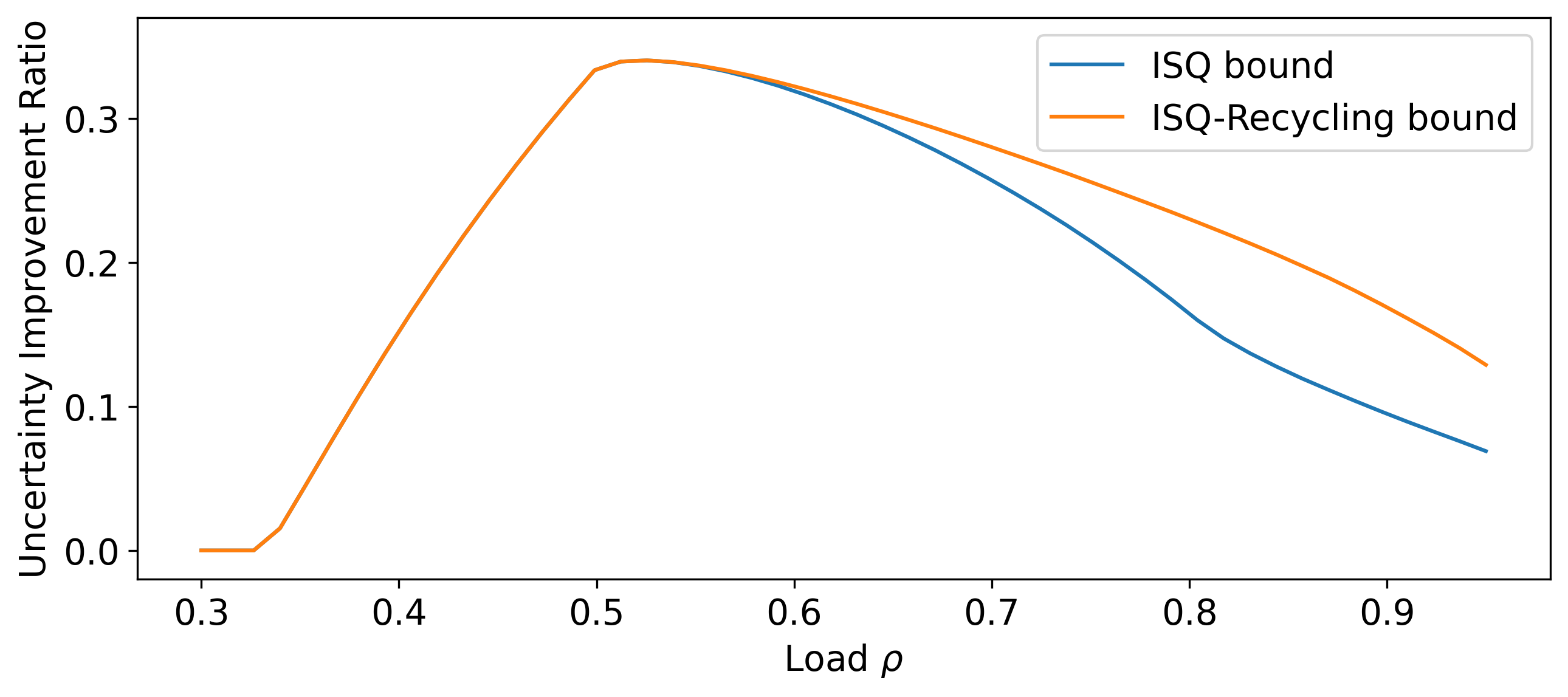}
    \caption{UIR relative to MixEx bound of the ISQ and ISQ-Recycling bound in the $M/G/2$ setting with $\text{Exp}(1)$ job sizes. See \Cref{sec:wine} for definitions of ISQ and ISQ-Recycling bounds.}
    
    \label{fig:intro_UIR}
\end{figure}

Our theoretical framework applies to any job size distribution, any number of servers $k\geq 2$, and any load in the $M/G/k$ system, as well as any scheduling policy.

\subsection{Outline of paper}
This paper is organized as follows:
\begin{itemize}
    \item \Cref{sec:prior_work}: We discuss prior work on multiserver scheduling and the drift method.
    \item \Cref{sec:model}: We introduce the Increasing Speed Queue (ISQ-$k$), related auxiliary systems and notation.
    \item \Cref{sec:main}: We present our main results and provide a proof outline.  
    \item \Cref{sec:wine}: We introduce the WINE-based framework for lower bounding the mean response time of the $M/G/k$ system and combine it with existing lower bounds.
    \item \Cref{sec:isq}: We establish results comparing the work in the ISQ-$k$ system to that in the $M/G/k$ system.
    \item \Cref{sec:drift}: We introduce the drift method and prove a version of the Basic Adjoint Relationship (BAR) applicable to our setting.
    \item \Cref{sec:mg2}: We start by proving our main results in the $k=2$ server case, lower bounding the $M/G/2$ under arbitrary scheduling policies before generalizing our results to the $k$-server setting.
    \item \Cref{sec:derive}: We introduce the DiffeDrift method, which we use to derive the ISQ-2 test functions, and generalize our results to the $k$ server system.
    \item \Cref{sec:mgk}: We prove our main results in the general $k\geq3$ case,  lower bounding $M/G/k$ mean response time under arbitrary scheduling policies.
    \item \Cref{sec:generalization}: We discuss how our framework can be extended to settings with unknown job sizes and to cases where job sizes are estimated. We also discuss computational considerations for large $k$.
    \item \Cref{sec:empirical}: We empirically demonstrate the tightness of our bounds via numerical calculation and simulation. 
\end{itemize}

\section{Prior Work}
\label{sec:prior_work}

\subsection{Shortest remaining processing time}
As discussed in the introduction, the Shortest-Remaining Processing Time (SRPT) policy is optimal in minimizing the mean response time in the $M/G/1$ queue (\citet{Schrage1968Proof} and \citet{Schrage1966}). In the multiserver setting, SRPT-$k$ is heavy-traffic optimal and provides an upper bound on the mean response time that is tight at very high loads (\citet{Grosof2018}).

At the opposite extreme, when there are no arrivals, it is well known that the Shortest Job First (SJF) policy is the optimal nonpreemptive policy in the multiserver setting. Furthermore, \citet{McNaughton1959} proves that no preemptive policy can outperform the optimal nonpreemptive policy. Since SRPT reduces to SJF when there are no arrivals, SRPT is also optimal in the multiserver setting.

However, SRPT-$k$ is not always the optimal policy across all load conditions. Notably, the SEK policy, recently introduced by \citet{Grosof2026}, empirically outperforms the response time of SRPT-$k$ under certain conditions. This raises the question of how much further improvement beyond SRPT-$k$ is possible.

The focus of this paper is on establishing lower bounds for the mean response time of $M/G/k$ systems under arbitrary policies, which has been noted as an open problem (\citet{Grosof2019Open}).

\subsection{Multiserver mean response time analysis}

The analysis of multiserver FCFS systems, particularly in the context of $M/G/k$-FCFS queues, is relatively well understood. Early work by \citet{Kingman1970} and \citet{Daley1998} provides bounds on the mean response time based on the first two moments of the job size distribution. \citet{Harchol-Balter2005} conduct an exact analysis of the mean response time in $M/Ph/k$-FCFS systems. Since any arbitrary distribution can be closely approximated by a phase-type distribution, their results offer a useful approximation for mean response time in $M/G/k$-FCFS systems. 

More recently, \citet{Gupta2011} demonstrate how utilizing the full sequence of moments of the job size distribution can yield 
bounds on mean response time. \citet{Li2024} extend Kingman’s bounds (\citet{Kingman1962}) to the multiserver $GI/GI/k$ queue and provide the first simple and explicit bounds for mean waiting time that scale with $1/(1-\rho)$.  These results allow for bounds on the mean response time in $M/G/k$-FCFS queues across all load levels. In contrast, our work considers a more general setting by allowing for any arbitrary scheduling policies and focuses on establishing lower bounds.

\subsection{Relevant work analysis}
\label{sec:relevant_recycling}

Relevant work is the amount of work in the system that is prioritised ahead of a given job. Relevant-work analysis has been widely used to study the response time of many scheduling policies, such as Preemptive Shortest Job First (PSJF) and Least Attained Service (LAS) (\citet{Harchol-Balter2013}).

Relevant work analysis can be applied to both single server and multiserver systems. The Work Integral Number Equality (WINE) formula, originally introduced by \citet{Scully2020} and discussed in detail in \Cref{sec:wine}, provides a new method for computing mean response time using SRPT-based relevant work. This formula holds for any queueing system under any scheduling policy not just the SRPT-based policy. Our approach builds on this framework, using WINE as the foundation for developing lower bounds on the mean response time.

Relevant work analysis was specifically developed for index policies. These are policies that assign each job a numerical index computed from the job’s own state, and select the job with the best index. \citet{Scully2018} introduce the terms ``recycling" and ``recycled jobs" to describe jobs whose original index exceeds a cutoff and later decreases below that cutoff, in order to keep track of more complex index functions such as the Gittins index. Since then, this concept and terminology have been used in several subsequent works, including the Gittins-$k$ policy (\citet{Scully2020}) and ServerFilling-SRPT (\citet{Grosof2023Optimal}). In our setting, we preserve this terminology and refer to jobs with an initial size larger than $x$ that have received service to bring their remaining size below $x$ as recycled jobs and such events as recyclings.

Finally, \citet{Scully2020} prove that the Gittins-$k$ policy is nearly optimal in the $M/G/k$ under extremely general conditions. The Gittins-$k$ policy coincides with SRPT-$k$ when there is perfect information about job sizes, which is the setting we study in this paper.

\subsection{Basic adjoint relationship (BAR) and drift method}
\label{sub_sec:prior_bar}

The well known Basic Adjoint Relationship (BAR) equation states that the stationary distribution $\pi$ of a continuous-time Markov chain $\{Z_t\}_{t\geq 0}$ with instantaneous generator $G$ satisfies
\begin{align}
\label{eqn:steady_state}
    \E_{\pi}[G\circ g(Z)]=0,
\end{align}
under suitable conditions on the Markov chain $\{Z_t\}_{t\geq 0}$ and the function $g$ (\citet{Glynn2008}). By carefully designing the test function $g$, one can solve the BAR equation \eqref{eqn:steady_state} to obtain either exact expressions for the moments of stationary variables or (asymptotically tight) bounds on these moments, depending on the system. This method is commonly referred to as the drift method (\citet{Eryilmaz2012,Grosof2023TRaM,Hong2023} and \citet{Maguluri2016}). We use the drift method to obtain the exact expected total work of the ISQ-$k$ system,
and exact results and bounds for further systems we introduce in \Cref{sec:model}. We use these exact results and bounds to lower bound mean response time in the $M/G/k$ under an arbitrary scheduling policy.
\subsubsection{Choosing test functions}

Recent applications of the drift method use simple exponential test functions of the form $e^{tQ}$, where $Q$ is the queue length, or $e^{tW}$, where $W$ is the total work in the system. For example, \citet{Braverman2017} and \citet{Braverman2024} use these test functions to prove heavy-traffic and steady-state approximations in both single-class and multi-class queueing systems. The use of exponential test functions is also known as the transform method (\citet{Hurtado-Lange2020} and \citet{jhunjhunwala2023}). 

However, these methods cannot be applied directly to queues with a variable work completion rate, since the test functions $e^{tQ}$ and $e^{tW}$ only yield useful information when the work completion rate is constant. Other recent works have focused on test functions of the form $Q^2$ or $W^2$ (\citet{Grosof2022MSJ}), which  likewise depend on a constant rate of work completion. Some attempts have been made to circumvent this problem by focusing on heavy traffic settings. For instance, \citet{Hurtado-Lange2020} rely heavily on state-space collapse to ensure a constant work completion rate except when the system is near-empty. As the system is rarely empty in heavy traffic, this state-space collapse allows the same straightforward test functions to be used, and heavy traffic results to be obtained.

However, in our intermediate-load setting, the variable work completion rate is fundamental, and we cannot eliminate that complexity from the problem. Instead, we introduce the DiffeDrift approach (see \Cref{sec:derive}), inventing a new class of test functions which can accommodate the ISQ-$k$ system's variable work completion rate.


\subsubsection{Sufficient conditions for BAR} Proposition 3 of \citet{Glynn2008} demonstrates a set of sufficient conditions for the BAR equation \eqref{eqn:steady_state} to hold. This proposition has been widely applied in many recent papers in queueing theory. Their result requires the state space to be discrete and certain regularity conditions to hold on the rate matrix of the Markov jump process and the test function $g$. For example, \citet{Grosof2023TRaM} and \citet{Hong2023} invoke Proposition 3 of \citet{Glynn2008} by showing that their respective Markov chains have uniformly bounded transition rates. However, this result does not directly apply to our system because our system has a continuous state space.

Theorem 2 of \citet{Glynn2008} provides another set of sufficient conditions for the BAR equation \eqref{eqn:steady_state} to hold. This result allows continuous state spaces but requires \emph{bounded-drift} test functions. However, it is not directly applicable to our unbounded quadratic test function, which has unbounded drift. 
For test functions with unbounded drift, a standard approach in the existing literature is to truncate these functions, thereby producing a sequence of bounded test functions that approximate the original (\citet{Braverman2017} and \citet{Guang2024}). 

We take an alternative approach, proving a new BAR result for time-homogeneous Markov processes with unbounded continuous state spaces and unbounded test functions that grow at most quadratically, see \Cref{sec:drift}. We find this approach much easier to apply than the standard truncation method given the complexity of our test functions arising from the DiffeDrift method.

\section{Model}
\label{sec:model}

In this section we specify our queueing model and introduce notation in \Cref{sec:notation}. We introduce a novel variable-speed queue called the Increasing Speed Queue (ISQ-$k$) in \Cref{subsec:isq}. Finally, we introduce two auxiliary queueing systems based on the ISQ-$k$ that we will use to lower bound the mean relevant work in the $M/G/k$, in \Cref{sec:aux_isq}. 

\subsection{Queueing model and notation}
\label{sec:notation}
We study lower bounds on the mean response time $\E[T^\pi]$ of the $M/G/k$ queue under an arbitrary Markovian scheduling policy $\pi$. Let $k$ denote the number of servers and $\lambda$ the arrival rate. A job's \emph{size} is the inherent amount of work in the job. Let jobs have i.i.d.\ sizes sampled from a job size distribution with probability density function (pdf) $f_S$ and cumulative distribution function (cdf) $F_S$. Let $S$ denote the corresponding random variable for job size and let $\widetilde{S}$ denote the Laplace–Stieltjes transform of $S$.
Note that we assume for simplicity that the job size distribution is continuous and has a pdf, but our results can be straightforwardly extended to a more general setting.
The service rate of each server is $1/k$, and the entire system has a maximum service rate of $1$, when all $k$ servers are occupied. In particular, a job of size $s$ will require a total service time of $ks$ to complete. We define the \textit{load} of the system to be the long-term fraction of servers which are in use. Load is given by $\rho=\lambda\E[S]$, and we assume $\rho < 1$ for stability.

We define a \emph{scheduling policy} $\pi$ to map the set of jobs currently in the $M/G/k$ system, as well as some auxiliary Markovian state,
to a choice of at most $k$ jobs to serve.
For example, a few common scheduling policies are First-Come First-Serve (FCFS-$k$), which serves the $k$ jobs which arrived longest ago, and Shortest-Remaining-Processing-Time (SRPT-$k$), which serves the $k$ jobs of least remaining size.

We say a job is \textit{relevant} at some threshold $x$
if its remaining size is below $x$.
This may occur if the job arrives into the system with an initial size $<x$, or if a job with an initial size $>x$ reaches a remaining size below $x$.
We refer to the process of a job receiving service and lowering its remaining size
as ``aging''.

We use $W$ to denote the total work in the system,
namely the sum of the remaining sizes of all jobs in the system.
We use $W_x$ to denote the total remaining size of all relevant jobs in the system, this is the SRPT \textit{relevant work}. However, our analysis applies to any arbitrary scheduling policy, not just those based on SRPT.

Let $\lambda_x:=\lambda F_S(x)$ denote the arrival rate of jobs which are relevant from the moment they enter the system.
Let $S_x:=[S\mid S\leq x]$ denote the conditional job size random variable,
truncated at a threshold $x$. The conditional load, $\rho_x$, at a threshold $x$ is the average rate at which works with initial size $<x$ arrives into the system is given by $\rho_x:=\lambda_x\E[S_x]=\lambda\E[S\mathbbm{1}_{\{S\leq x\}}]$.

We also define the capped job size random variable, $S_{\bar{x}} := [\min\{S, x\}]$, which represents the amount of relevant work at some threshold $x$ that a random job will contribute over its time in the system. The capped load, $\rho_{\bar{x}}$, is similarly given by $\rho_{\bar{x}}=\lambda\E[S_{\bar{x}}]=\lambda\E[\min\{S,x\}]$.

\subsection{Increasing speed queue}
\label{subsec:isq}

The increasing-speed queue (ISQ-$k$) is a single-server, variable speed queue. The speed of the ISQ-$k$ system is defined as follows: when the first job arrives to an empty queue, the server initially runs at speed $1/k$. If another job arrives before the system empties, the server now runs at speed $2/k$. With each arrival during a busy period, the server's speed increases by $1/k$ until it reaches the maximum speed of 1. The server maintains the maximum speed of 1 until the system empties and a busy period ends, resetting the speed to 0.
The state of the ISQ-$k$ system is given by a pair $(w,i)$, where $w$ denotes the work in the system and $i$ denotes the speed of the system. In particular, $i\in\{0,1/k,2/k,\cdots,1\}$. We write $W(t), I(t)$ to denote the state of the system at a particular time $t$, and we write $W, I$ to denote the corresponding (correlated) stationary random variables.

\subsection{Auxiliary ISQ systems and recycling}
\label{sec:aux_isq}
Recall that we call a job relevant if its remaining size is less than some threshold $x$. There are two ways for a job to be relevant: it either enters the system with size less than $x$, or it ages down from an initial size larger than $x$ to a remaining size less than $x$.
We refer to the event in which a job with initial size exceeding 
$x$ ages down to a remaining size below $x$ as \textit{recycling}, and we call any such job a \textit{recycled} job. This is standard terminology and see \Cref{sec:relevant_recycling} for how prior work uses the recycling terminology.

We define two auxiliary systems to help us lower bound the response time of the $M/G/k$ system: the \textit{separate ISQ-$k$} and the \textit{recycling ISQ-$k$}.

A separate ISQ-$k$ system (Sep-ISQ-$k$) consists of two subsystems, the truncated-ISQ-$k$ subsystem and the $M/G/\infty$ subsystem. Jobs arrive to the overall Sep-ISQ-$k$ system with the same arrival process as the $M/G/k$, jobs arriving into the system according to a Poisson process with rate $\lambda$ and with job size $S$. Jobs with initial sizes $\leq x$ are routed to an ISQ-$k$ subsystem, which we call the truncated-ISQ-$k$ subsystem. Specifically, the truncated-ISQ-$k$ system is an ISQ-$k$ system characterized by an arrival rate $\lambda_x = \lambda F_S(x)$, and job size $S_x \sim [S \mid S \leq x ]$.

On the other hand, jobs with initial sizes $> x$ are routed to the separate $M/G/\infty$ subsystem, each server operating at a speed of $1/k$. That is, jobs which are not relevant on arrival are served at a separate subsystem.

The recycling ISQ-$k$ system (Rec-ISQ-$k$) is an ISQ-$k$ system with two arrival streams, a Poisson arrival process with a rate of $\lambda F_S(x)$ and job size $S_x \sim [S\mid S\leq x]$, and a general arrival process $\{R_t\}$ with rate $\lambda(1-F_S(x))$ and job sizes exactly $x$.
The general arrival process is governed by some general Markovian process $\{R_t\}$ for which we only require stability and can be causally dependent on the ISQ-$k$ system, we denote the state of this process as $r$. The idea behind this general arrival process is that it models the moments at which jobs recycle in the M/G/$k$ system under a general Markovian scheduling policy.  We refer to the Poisson arrivals as the \emph{truncated stream} and the general arrival process as the \emph{recycling stream}.

We will use the triplet $(w, i, r)$ to denote a state of the recycling system where $w$ and $i$ are defined similarly to the states of the full ISQ-$k$ system and $r$ denotes a separate Markovian state from the recycling stream's state space, which incorporates additional information about the recycling stream. As before, $R$ denotes the corresponding stationary random variable.

In \Cref{thm:isq_sep}, we show that the mean relevant work in the Sep-ISQ-$k$ system lower bounds the mean relevant work in the $M/G/k$ system under any arbitrary scheduling policy, for any arrival rate $\lambda$, job size $S$, and relevancy cutoff $x$. In \Cref{thm:ar_isq}, we prove there always exists a Rec-ISQ-$k$ system lower bounding the mean relevant work of a $M/G/k$ system under an arbitrary scheduling policy. Therefore, our lower bound on the mean relevant work of Rec-ISQ-$k$ system with any recycling process $R$ must lower bound the mean relevant work of any $M/G/k$ system.

\section{Main Results}
\label{sec:main}
In this paper, we give the first nontrivial lower bounds on the mean response time of the $M/G/k$ system under an arbitrary scheduling policy. Our bounds are based on lower bounds on relevant work, which we can turn into bounds on response time via the WINE formula (\Cref{prop:wine}). Directly applying this gives the first novel lower bound, the MixEx bound. We improve upon these by proving stronger relevant work lower bounds with the ISQ-$k$ system, these give even stronger bounds, the ISQ and ISQ-Recycling bounds.

We start by lower bounding the 2-server system before moving on to the $k$-server system. For the 2-server system, these lower bounds on mean response time are as follows:

\begin{theorem}
\label{thm:main_2}
The mean response time of the optimal $M/G/2$ system is lower bounded as follows, which we refer to as the MixEx, ISQ, and ISQ-Recycling bounds, respectively:
\begin{align}
\label{eqn:main_mixex2}
E[T^{M/G/2/OPT}]&\geq\dfrac{1}{\lambda}\int_0^{\infty}\dfrac{\max\{\E[W_x^{M/G/1\text{-SRPT}}],\E[W_x^{M/G/\infty}]\}}{x^2}\,dx,\\
\label{eqn:main_isqs2}
E[T^{M/G/2/OPT}]&\geq\dfrac{1}{\lambda}\int_0^{\infty}\dfrac{\max\{\E[W_x^{M/G/1\text{-SRPT}}],\E[W_x^{M/G/\infty}],\E[W_x^{\text{Sep-ISQ-}2}]\}}{x^2}\,dx,\\
\label{eqn:main_isqr2}
E[T^{M/G/2/OPT}]&\geq\dfrac{1}{\lambda}\int_0^{\infty}\dfrac{\max\{\E[W_x^{M/G/1\text{-SRPT}}],\E[W_x^{M/G/\infty}],\E[W_x^{\text{Sep-ISQ-}2}],\E[W_x^{\text{Rec-ISQ-}2\text{-L}}]\}}{x^2}\,dx,
\end{align}
where $\E[W_x^{M/G/1\text{-SRPT}}]$ and $\E[W_x^{M/G/\infty}]$ are standard results given in \Cref{sec:wine} and
\begin{align}
\label{eqn:main_2_1}
\E[W_x^{{\text{Sep-ISQ-}2}}]&=\dfrac{\lambda_x\E[S_x^2]}{2(1-\rho_x)}+\dfrac{\E[S_x]-(1-\widetilde{S_x}(2\lambda_x))/2\lambda_x}{3-\widetilde{S_x}(2\lambda_x)}+(\lambda-\lambda_x)x^2, \\
\label{eqn:main_2_2}
\E[W_x^{\text{Rec-ISQ-}2\text{-L}}]&= \dfrac{\lambda_x \E[S_x^2]}{2 (1 - \rho_x)} + \dfrac{\E[S_x]-(1-\widetilde{S_x}(2\lambda_x))/2\lambda_x}{3-\widetilde{S_x}(2\lambda_x)}\cdot\dfrac{1-\rho_{\overbar{x}}}{1-\rho_x}+\dfrac{(\lambda - \lambda_x)x^2}{2(1-\rho_x)}.
\end{align}
\end{theorem}

\begin{proof}[Proof deferred to \Cref{sec:thm_main_2}]
\end{proof}

We call \Cref{eqn:main_mixex2} the \emph{MixEx bound}, as it combines two existing lower bounds on mean relevant work, namely the $M/G/1$-SRPT and $M/G/\infty$ relevant work, through the WINE formula. 
These existing lower bounds, together with the WINE formula, are explained in greater detail in \Cref{sec:wine}.  

The remaining parts of \Cref{thm:main_2} are proved in \Cref{sec:thm_main_2}, see \Cref{sec:proof_outline} for proof outline. 
We refer to \Cref{eqn:main_isqs2} as the \emph{ISQ bound} and to \Cref{eqn:main_isqr2} as the \emph{ISQ-Recycling bound}, because each introduces a new lower bound on the mean relevant work obtained from the Sep-ISQ-2 and Rec-ISQ-2 systems, respectively. Specifically, we derive \Cref{eqn:main_2_1,eqn:main_2_2} by applying our novel variant of the drift method, called the DiffeDrift method, to the Sep-ISQ-2 and Rec-ISQ-2 systems, respectively. 
These ISQ-2 variants are introduced in \Cref{sec:aux_isq}. Numerical studies of these bounds are given in \Cref{sec:empirical}.

We now state our analogous result for the more general $k$-server system, which we prove in \Cref{sec:mgk}.

\begin{theorem}
\label{thm:main_k}
The mean response time of the optimal $M/G/k$ system with $k\geq 3$ is lower bounded as follows, which we refer to as the MixEx, ISQ, and ISQ-Recycling bounds, respectively:
\begin{align}
\label{eqn:main_mixex_k}
E[T^{M/G/k/OPT}]&\geq\dfrac{1}{\lambda}\int_0^{\infty}\dfrac{\max\{\E[W_x^{M/G/1\text{-SRPT}}],\E[W_x^{M/G/\infty}]\}}{x^2}\,dx,\\
\label{eqn:main_isqsk}
E[T^{M/G/k/OPT}]&\geq\dfrac{1}{\lambda}\int_0^{\infty}\dfrac{\max\{\E[W_x^{M/G/1\text{-SRPT}}],\E[W_x^{M/G/\infty}],\E[W_x^{\text{Sep-ISQ-}k}]\}}{x^2}\,dx,\\
\label{eqn:main_isqrk}
E[T^{M/G/k/OPT}]&\geq\dfrac{1}{\lambda}\int_0^{\infty}\dfrac{\max\{\E[W_x^{M/G/1\text{-SRPT}}],\E[W_x^{M/G/\infty}],\E[W_x^{\text{Sep-ISQ-}k}],\E[W_x^{\text{Rec-ISQ-}k\text{-L}}]\}}{x^2}\,dx,
\end{align}
where $\E[W_x^{M/G/1\text{-SRPT}}]$ and $\E[W_x^{M/G/\infty}]$ are standard results given in \Cref{sec:wine} and
\begin{align}
\label{eqn:main_k_1}
\E[W_x^{\text{Sep-ISQ-}k}]&= \dfrac{\lambda_x\E[S_x^2]}{2(1-\rho_x)}+\dfrac{\lambda_x\E[v_1(S_x)]}{2+2\lambda_x\E[u_1(S_x)]} +\dfrac{k(\lambda-\lambda_x)x^2}{2},\\
\label{eqn:main_k_2}
\E[W_x^{\text{Rec-ISQ-}k\text{-L}}]&= \dfrac{\lambda_x\E[S_x^2]}{2(1-\rho_x)}+\dfrac{\lambda_x\E[v_1(S_x)]}{2+2\lambda_x\E[u_1(S_x)]}\cdot\dfrac{1-\rho_{\bar{x}}}{1-\rho_x}+\frac{(\lambda - \lambda_x)J_x}{2(1-\rho_x)},
\end{align}
where $J_x:=\min\left\{x^2, \min_{i < 1}\left\{
        \inf_{w\in[0,kix]}h_{k,x}(w+x,i+1/k)-h_{k,x}(w,i)
        \right\}\right\}$.
\end{theorem}
\begin{proof}[Proof deferred to \Cref{sec:thm_main_k}]
\end{proof}

In the above theorem, we define $u_1$ and $v_1$ via the following recursive formulas,
\begin{align*}
    u_{q}(w)&=e^{-\frac{kw\lambda_x}{q}}\int_0^w\dfrac{e^{\frac{k\lambda_x y}{q}}(k-q+k\lambda_x \E[u_{q+1}(S_x+y)])}{q}\,dy,\\
    v_q(w)&=e^{-\frac{kw\lambda_x}{q}}\int_0^w \dfrac{e^{\frac{k\lambda_x y}{q}}(2ky-2qy+k\lambda_x \E[v_{q+1}(S_x+y)])}{q}\,dy,
\end{align*} 
with the initial condition $u_k(w)=0$ and $v_k(w)=0$ (see \Cref{def:isqk_contant,def:isqk_affine}). Both $u_1$ and $v_1$ are also derived using the DiffeDrift method. The function $h_{k,x}$ is defined in \Cref{def:isqk_modified}. \Cref{eqn:main_k_1} is based on the Sep-ISQ-$k$ system and \Cref{eqn:main_k_2} is based on the Rec-ISQ-$k$ system.

Note that $u_1, v_1,$ and $h_{k,x}$ can all be exactly symbolically derived for any number of servers, see \Cref{apdx:isq3} for $u_1$ and $v_1$ under the 3,4 and 5-server cases. See \Cref{sec:wine} for how to convert lower bounds on mean relevant work into lower bounds on mean response time. In \Cref{sec:generalization}, we further discuss how a specialized program can be developed to compute these recursive formulas efficiently for $k \geq 6$, and how our framework extends beyond the known-size setting.

\subsection{Proof outline}
\label{sec:proof_outline}
To derive lower bounds on mean response time, we first obtain lower bounds on the mean relevant work. By the WINE formula (see \Cref{sec:wine}), these directly translate into lower bounds on mean response time. 
Applying this framework to the naive lower bounds (\Cref{eqn:lower_1,eqn:lower_2}) already yields some improvement in certain load regimes. Our approach, however, goes further, providing substantially tighter lower bounds.

We derive the novel relevant work lower bounds given in \Cref{thm:main_2,thm:main_k}, via two steps: First, we prove that the expected relevant work of each ISQ-$k$ system lower bounds the optimal expected relevant work in the $M/G/k$ under any scheduling policy. Second, we characterize the work in the ISQ-$k$ system by applying our novel DiffeDrift method. 

\subsubsection*{Total work lower bound}

Our first step is to prove that the \emph{total work} in the ISQ-$k$ system lower bounds $M/G/k$ \emph{total work} under an arbitrary scheduling policy:

\begin{restatable}{theorem}{mainfive}
\label{thm:main}
For any job size $S$ and arrival rate $\lambda$, the expected total work in the $M/G/k$ system under any scheduling policy is lower bounded by the expected total work in the ISQ-$k$ system.
\end{restatable}

We prove \Cref{thm:main} using a sample path coupling argument in \Cref{sec:isq}.

\subsubsection*{Relevant work lower bound}

Using \Cref{thm:main}, we prove in \Cref{thm:isq_sep} that the expected \textit{relevant work} of the Sep-ISQ-$k$ system lower bounds the $M/G/k$ expected \textit{relevant work} under any scheduling policy for all arrival rates $\lambda$. We also prove in \Cref{thm:ar_isq} that for any $M/G/k$ scheduling policy there always exists a recycling stream $\{R_t\}$ such that expected \textit{relevant work} of the Rec-ISQ-$k$ system lower bounds that policy's $M/G/k$ expected \textit{relevant work} for all arrival rates $\lambda$. 

\subsubsection*{Characterizing ISQ-$k$ total work}

Now that we've proven that we can use Sep-ISQ-$k$ and Rec-ISQ-$k$ to lower bound mean relevant work in the M/G/$k$, our remaining goal is to characterize the expected relevant work of the Sep-ISQ-$k$ and Rec-ISQ-$k$ systems. As an intermediate result, a key step is to compute the expected total work of the ISQ-$k$ system. We analyze the expected total work of the ISQ-$k$ system using the DiffeDrift method, see \Cref{sec:derive_affine}.

The expected total work of the ISQ-2 and ISQ-$k$ systems are given by:
\begin{align}
\label{eqn:main_isq2}
        \E[W^{\text{ISQ-2}}]&=\dfrac{\lambda\E[S^2]}{2(1-\lambda\E[S])}+\dfrac{\E[S]-(1-\widetilde{S}(2\lambda))/2\lambda}{3-\widetilde{S}(2\lambda)},\\
        \label{eqn:main_isqk}
        \E[W^{\text{ISQ-$k$}}]&=\dfrac{\lambda\E[S^2]}{2(1-\lambda\E[S])}+\dfrac{\lambda\E[v_1(S)]}{2+2\lambda\E[u_1(S)]}.
    \end{align} 
Here $v_1$ and $u_1$ are functions defined  as the solutions of differential equations in \Cref{def:isqk_contant,def:isqk_affine}.
\Cref{eqn:main_isq2,eqn:main_isqk} correspond to \Cref{prop:isq2_work,prop:isqk_work}.
These results are proved in \Cref{sec:mg2,sec:mgk}.

\subsubsection*{Sep-ISQ-$k$ and Rec-ISQ-$k$ relevant work}
Now, it remains to analyze the mean relevant work in Sep-ISQ-$k$ and Rec-ISQ-$k$. We exactly characterize the relevant work of the Sep-ISQ-$k$ system and derive a lower bound on the relevant work of the Rec-ISQ-$k$ system.

We characterize the exact expected relevant work of the Sep-ISQ-$k$ system in \Cref{thm:sep_2,thm:sep_k}.
As a result, we lower bound the expected relevant work of any $M/G/k$ system under any scheduling policies. In our main results, \Cref{thm:main_2,thm:main_k}, our \Cref{eqn:main_2_1,eqn:main_k_1} are based on the Sep-ISQ-$k$ results.

For the Rec-ISQ-$k$ system, recycling presents multiple difficulties. First, the ISQ-$k$ total work formulas \eqref{eqn:main_isq2} and \eqref{eqn:main_isqk} cannot be applied directly, so we require specialized test functions, see \Cref{sec:derive_modified}. 
Additionally, the changes in the drift caused by the recycling events are difficult to characterize exactly. Therefore, we uniformly lower bound the jumps of the test function during these events, allowing us to lower bound the expected relevant work of any Rec-ISQ-$k$ system with any arbitrary recycling stream $\{R_t\}$. These results are proved in \Cref{thm:ar_2,thm:ar_k}. By \Cref{thm:ar_isq}, we therefore lower bound the expected relevant work of any $M/G/k$ system under arbitrary scheduling policies, resulting in \Cref{eqn:main_2_2,eqn:main_k_2} in our main results \Cref{thm:main_2,thm:main_k}.

\section{WINE and Existing Lower Bounds}
\label{sec:wine}

To translate lower bounds on the mean relevant work of the $M/G/k$ system into lower bounds on its mean response time under arbitrary scheduling policies, we use the Work Integral Number Equality (WINE) introduced in Theorem 6.3 of \citet{Scully2020}. WINE allows us to characterize the mean response time of a generic queueing system under a generic scheduling policy in terms of the mean relevant work for each relevancy-cutoff level $x$.

\begin{proposition}[WINE Identity]
\label{prop:wine}
    For an arbitrary scheduling policy $\pi$, in an arbitrary system,
    \begin{align}
    \label{eqn:wine}
        \E[T^{\pi}]=\dfrac{1}{\lambda}\E[N^{\pi}]=\dfrac{1}{\lambda}\int_0^{\infty}\dfrac{\E[W_x^{\pi}]}{x^2}\,dx.
    \end{align}
\end{proposition}

WINE has been used to upper bound mean response time under complex scheduling policies in the $M/G/1$ (\citet{Scully2022Uniform}) as well as in multiserver systems (\citet{Grosof2023Optimal} and \citet{Scully2020}).

We use \Cref{prop:wine} in a novel fashion: rather than upper bounding relevant work for a specific policy to upper bound response time for that policy, we instead lower bound mean relevant work over all policies to lower bound mean response time over all policies.
In particular, if we can prove relevant work lower bounds at different relevancy cutoffs $x$ using different proof methods,
we can combine them by taking the strongest bound for every $x$ and integrating with \Cref{prop:wine}. This will provide a stronger bound than any prior individual response time lower bound.


Previously, only two naive lower bounds on the mean response time and mean relevant work of the $M/G/k$ under arbitrary scheduling policies have appeared in the literature. 

The first lower bound is given by the resource-pooled $M/G/1$ queue under the single-server optimal SRPT policy. The optimal policy in the resource-pooled $M/G/1$ must lower bound the optimal policy in the $M/G/k$ system, as the resource-pooled model can simulate any $M/G/k$ policy. Since we are using the SRPT relevant work, these are precisely the jobs that SRPT prioritizes. Therefore, the resource-pooled single server SRPT provides a lower bound on the relevant work of any $M/G/k$ system under arbitrary scheduling policy.

The mean relevant work in the SRPT-1 system was characterized by \citet{Schrage1966} and is given by:
\begin{align}
\label{eqn:lower_1}
    \E[W_x^{M/G/k}]\geq \E[W_x^{M/G/1 \text{-SRPT}}] = \dfrac{\lambda}{2}\dfrac{\E[S_{\bar{x}}^2]}{1-\rho_x}.
\end{align}

The second bound is the mean service time bound. To compute the mean relevant work, we imagine jobs arriving into an $M/G/\infty$ queue, whose servers run at the same speed as the $M/G/k$ servers. Each job of initial size $s$ will spend $k \min\{s, x\}$ amount of time as a relevant job, as its size shrinks linearly down to $0$. The job thus contribute $\frac{k}{2}\min\{s, x\}^2$ amount of relevant work. Therefore,
\begin{align}
\label{eqn:lower_2}
    \E[W_x^{M/G/k}]\geq \E[W_x^{M/G/\infty}]=\dfrac{k\lambda}{2}\E[S_{\bar{x}}^2].
\end{align}
The above bound is tight at low load but loose at medium and high load, see \Cref{fig:intro_wine}.

Taking the maximum of these two bounds for each $x$ and integrating using WINE (\Cref{prop:wine}) gives us our first improved lower bound on the mean response time of the $M/G/k$ under an arbitrary scheduling policy, which we call the MixEx bound, i.e.,  \Cref{eqn:main_mixex2} and \Cref{eqn:main_mixex_k}. 

The primary goal of this paper is to further improve upon MixEx by incorporating the exact mean relevant work of the Sep-ISQ-$k$ system, giving an improved lower bound on mean relevant work in the $M/G/k$, namely \Cref{eqn:main_k_1}. Combining these bounds with WINE results in a new lower bound which we call the ISQ bound. Finally, the ISQ bound is augmented with a lower bound on the mean relevant work of the Rec-ISQ-$k$ system, giving a final lower bound on the $M/G/k$, namely \Cref{eqn:main_k_2}, resulting in the final ISQ-Recycling bound.

\section{Increasing Speed Queue}
\label{sec:isq}
In this section, we prove that the Sep-ISQ-$k$ system and the Rec-ISQ-$k$ system lower bound the relevant work of an $M/G/k$ system under an arbitrary scheduling policy.

We first show that, for any sequence of arrival times and job sizes,
the total work in a $k$-server system under an arbitrary scheduling policy is lower bounded by the total work in the ISQ-$k$ system, under the same sequence of arrival times and job sizes.
\begin{proposition}
\label{prop:work}
    For an arbitrary sequence of arrival times and job sizes and at any given point in time $t$, the total work in a $k$-server system under an arbitrary scheduling policy is lower bounded by the total work in an ISQ-$k$ system with the same arrival sequence.
\end{proposition}

\begin{proof}
    We will use the index $j\in \mathcal{J}=\{1,2,...\}$ to denote the busy periods of the ISQ-$k$ system. Specifically, a busy period begins when a job arrives to an empty system, and ends when the system is next empty.
    Let $W_j^{ISQ{\text -}k}(t)$ denote the total work in the ISQ-$k$ system consisting of jobs which arrived during busy period $j$ of the ISQ-$k$ system.
    Note that at all points in time, $W_j^{ISQ{\text -}k}(t)$ is positive for at most one value of $j$, namely the current busy period.
    Similarly, let $W_j^k(t)$ denote the total work remaining in the $k$-server system consisting of jobs which arrived during busy period $j$ of the ISQ-$k$ system.

    Note that the index $j$ always refers to busy periods of the ISQ-$k$ system -- we ignore busy periods of the $k$-server system. We will show that $W_{j}^{ISQ{\text -}k}(t)\leq W_j^{k}(t)$ for all $j\in \mathcal{J}$ and for all $t\geq 0$, which suffices to bound overall work at time $t$. Let  $A_j(t)$ denote the number of jobs which have arrived during busy period $j$ by time $t$. Let $B_j^{k}(t)$ denote the fraction of servers in the $k$-server system which are allocated to jobs which arrived during busy period $j$.
    
    Because arrivals to both systems are identical,
    it suffices to show that the rate of completion of $W_{j}^{ISQ{\text -}k}(t)$ exceeds that of $W_j^{k}(t)$ at all times $t$ at which $W_j^{ISQ{\text -}k}(t) > 0$. Note that the rate of completion of $W_j^{k}(t)$ is simply $B_j^{k}(t)$,
    thanks to our definition that each server in the $M/G/k$ completes work at rate $1/k$. Define $B_j^{ISQ{\text -}k}$ similarly.

    Thus, we want to show that $B_j^{ISQ{\text -}k}(t)$ is higher than $B_j^{k}(t)$ whenever $W_j^{ISQ{\text -}k}>0$. Note that $W_j^{ISQ\text{-}k}>0$ only occurs when $j$ is the busy period currently active.
    Note that $W_j^{ISQ{\text -}k}(t) = W^{ISQ{\text -}k}(t)$ where $j$ is the current busy period at time $t$,
    as all of the work in the ISQ-$k$ must have arrived during its current busy period,
    and the ISQ-$k$ system ends each busy period by completing all work.
    Thus, $B_j^{ISQ{\text -}k}(t)$ is  $\min\{A_j(t)/k, 1\}$.
        
    $B_j^{k}(t)$ is bounded above by $\min\{A_j(t)/k,1\}$, based on the number of jobs which have arrived during this busy period, namely $A_j(t)$. This is at most the completion rate in the ISQ-k system, as desired. This completes the proof.

\end{proof}

An immediate consequence of \Cref{prop:work} is \Cref{thm:main}, which states that for a Poisson arrival process, the expected work of the ISQ-$k$ system must lower bound that of the $M/G/k$ system. In \Cref{prop:isqk_work} of \Cref{sec:mgk} we provide an exact formula for the mean relevant work of the ISQ-$k$ system.

\mainfive*

Using \Cref{thm:main}, we prove \Cref{thm:isq_sep}, which states that the mean work in the Sep-ISQ-$k$ system lower bounds the mean relevant work of the $M/G/k$ system under an arbitrary scheduling policy. In \Cref{thm:sep_k} of \Cref{sec:mgk}, we exactly characterize the mean relevant work for the Sep-ISQ-$k$ system.

\begin{theorem}
\label{thm:isq_sep}
     For an arbitrary job size $S$ and arrival rate $\lambda$, the expected relevant work in the $M/G/k$ system under an arbitrary scheduling policy is lower bounded by the expected relevant work in the separate-ISQ-$k$ system.
\end{theorem}
\begin{proof}
    We want to lower bound the relevant work in the $M/G/k$. We will divide that work into two categories: Relevant work from jobs with original size $\leq x$, and relevant work from jobs with original size $> x$. We will show that the truncated-ISQ-$k$ system lower bounds the former category, and the M/G/$\infty$ with server speed $1/k$ lower bounds the latter category.

    We prove this theorem using a coupling argument. By coupling, we mean that we construct two different queueing systems on the same probability space and let them evolve in parallel under the exact same arrival process. In particular, every job arrival occurs at the same time in both systems and has the same job size in each system. The only difference between the systems lies in how they process these arrivals.
    
    Consider a coupled pair of systems: a truncated $M/G/k$ system and the full $M/G/k$ system. Whenever a job arrives to the full $M/G/k$, if that job has size $\le x$, a job with the same size arrives to the truncated $M/G/k$. By Poisson splitting, the arrival process to the truncated $M/G/k$ is the desired process.
    
    To complete the coupling, we need to specify how the scheduling policies in the two systems relate to each other.
    Let an arbitrary scheduling policy be used in the full $M/G/k$. At any given point in time, we specify that the truncated $M/G/k$ serves each job with original size $\leq x$ that is in service in the full $M/G/k$ at that point in time. Note that this policy may waste servers, but it is an admissible policy.
    
    With this coupling in place,
    the total work in the truncated-$M/G/k$ system with this scheduling policy, all of which is relevant, lower bounds the total work from jobs with initial size $\leq x$ in the full-$M/G/k$ system at any given time, all of which is similarly relevant.
    
    From \Cref{thm:main}, we know that the total work in the truncated $M/G/k$ is lower bounded by the total work in the truncated-ISQ-$k$ system. This completes the bound for jobs with original size $\le x$.

    For jobs with initial size $> x$, the expected relevant work in the separate $M/G/\infty$ system lower bounds the expected relevant work of jobs with initial size $>x$ in the full-$M/G/k$ system because the rate of completion in the $M/G/\infty$ system is always $1/k$ for each job whereas the rate of completion of each job in the full $M/G/k$ is never more than $1/k$.
    
    Thus we see that the expected relevant work in the full-$M/G/k$ system is lower bounded by the expected relevant work in the separate-ISQ-$k$ system.

\end{proof}

We now switch our focus to the recycling ISQ-$k$ system.
Recall that we define a recycling to occur when a job of original size greater than $x$ ages down to remaining size of exactly $x$ in the $M/G/k$. From the perspective of relevant work, a recycling event looks like a remaining-size-$x$ job popping into existence, i.e. becoming relevant for the first time. Recyclings happen with rate $\lambda (1-F_S(x))$, because every job with original size $> x$ eventually recycles.

In particular, imagine a job of size $x$ arriving into the Rec-ISQ-$k$ system whenever a recycling occurs in the $M/G/k$ system. This arrival sequence has a rate $\lambda (1-F_S(x))$ and is Markovian.

\begin{theorem}
\label{thm:ar_isq}
    For an arbitrary job size $S$ and arrival rate $\lambda$, and an arbitrary $M/G/k$ scheduling policy,
    there exists a Markovian recycling stream $\{R_t\}$ such that
    the expected relevant work in the $M/G/k$ under the given scheduling policy is lower bounded by the expected relevant work in an Rec-ISQ-$k$ with the given recycling stream. 
\end{theorem}
\begin{proof}
Consider the full $M/G/k$, and specifically consider the relevant work in the $M/G/k$.

Consider a specific realization of the Poisson arrival sequence into the full $M/G/k$ under an arbitrary scheduling policy. For jobs with original size greater than $x$, there is some time  when the job becomes relevant for the first time. Therefore, from the perspective of the relevant work in the $M/G/k$ this is equivalent to a job with size exactly $x$ arriving into the system according to some arbitrary arrival process. 

In other words, the relevant work in the full $M/G/k$ matches the total work in a $k$-server system with two arrival streams: Poisson arrivals for jobs with size $\le x$, and recycling arrivals at rate $\lambda(1-F_S(x))$ of size exactly $x$, with equivalent scheduling policies.

For this specific arbitrary arrival stream, by  \Cref{prop:work}, we have that the relevant work in the Rec-ISQ-$k$ system with the same recycling stream lower bounds the relevant work in the arbitrary $k$-server system and thus lower bounds the relevant work in the full $M/G/k$ system. 

\end{proof}

Therefore, by \Cref{thm:ar_isq}, the minimum possible mean relevant work in the $M/G/k$ system under an arbitrary scheduling policy is lower bounded by the minimum possible mean work in the Rec-ISQ-$k$ system under an arbitrary recycling stream. We lower bound the mean work in the Rec-ISQ-$k$ system in \Cref{thm:ar_2} for the 2 server case and \Cref{thm:ar_k} for the general case.

\section{Drift Method}
\label{sec:drift}

In this section, we discuss the drift method/BAR approach, which we will use to analyze the ISQ-$k$ system.
We provide background on the drift method in \Cref{sec:drift_background}.
We then prove a novel instance of the BAR,  which holds for a class of unbounded test functions and Markov processes with continuous state spaces in \Cref{sec:novel_bar}.
Finally, we apply the drift method and our BAR result in particular to the ISQ-$k$ system in \Cref{sec:drift_isq}.

\subsection{Background on drift method}
\label{sec:drift_background}
The drift method (\citet{Eryilmaz2012}), also known as the BAR (\citet{Braverman2017} and \citet{Harrison1987}) or the rate conservation law (\citet{Miyazawa1994}),
states that the average rate of increase and decrease of a random variable must be equal, for any random variable with finite expectation and satisfying certain regularity conditions.

To formalize this concept, we make use of the \textit{drift} of a random variable. The drift is the random variable's instantaneous rate of change, taken in expectation over the system's randomness. Let $G$ denote the \textit{instantaneous generator}, which acts as the stochastic counterpart to the derivative operator. We can also apply $G$ to functions of the system state, which implicitly define random variables. We call such functions \textit{test functions}: a function $f$ that maps system states $(w,i)$ to real values. The instantaneous generator takes $f$ and outputs $G\circ f$, the drift of $f$.

Let $G$ denote the generator operator for the ISQ-$k$ system. Recall that $W(t)$ denotes the work in the system at time $t$ and $I(t)$ denotes the speed of the system at time $t$. For any test function $f(w,i)$,
\begin{align*}
    G \circ f(w,i):=\lim_{t\to 0}\dfrac{1}{t}\E[f(W(t),I(t))-f(w,i)\,|\,W(0)=w, I(0)=i].
\end{align*}

We usually do not directly use the above equation to compute the drift of a test function when we know all the rates at which the system states change. For the ISQ-$k$ system, $w$ increases according to stochastic jumps of size $S$, which arrive according to a Poisson Process with rate $\lambda$. These jumps also cause $i$ to increase by $1/k$ as long as $i < 1$. When $i>0$, $w$ decreases at rate $i$ due to work completion. The next lemma is a special case of Section 6.1 in \citet{Braverman2024}, which states that for a piecewise deterministic Markov process, the drift can be characterized by jumps caused by a Poisson process.

\begin{lemma}
\label{lem:isq_drift}
    For any real-valued differentiable function $g$ of the state of the ISQ-$k$ system,
    \begin{align}
        G\circ g(w,i)=\lambda\E[g(w+S,\min\{1,i+1/k\})-g(w,i)]-\frac{d}{dw}g(w,i)i.
    \end{align}
\end{lemma}

A key fact about drift, known as the Basic Adjoint Relationship (BAR) \eqref{eqn:steady_state}, is that in a stationary system, the expected drift of any random variable or test function is zero, as long as the random variable or test function has finite expectation in stationarity, and satisfies certainty regularity conditions. 

Existing BAR results unfortunately do not cover the ISQ-$k$ system and the specific test functions that we will use to characterize mean workload.
We therefore prove a novel basic adjoint relationship in that applies to our system and test functions. Specifically, our new results, \Cref{prop:drift} and \Cref{lem:uniform}, cover a class of unbounded test functions in our continuous-state setting which were not previously covered by existing results. An alternative approach would be to truncate the test functions and take a limit of a series of truncated test functions, but we expect this approach to be significantly more complicated in our setting. 

\subsection{BAR for unbounded test functions of continuous Markov processes}
\label{sec:novel_bar}

We begin by establishing the general framework that serves as the starting point for our BAR.

Let $(X_t)_{t\geq 0}$ be a continuous-time, time-homogeneous Markov process with transition probability kernel $P(t, x, U)$.  We denote the state space as $\mathcal{S}$, and we emphasize that it can be continuous and unbounded.

We denote the generator as $G$,
\begin{align*}
    G \circ g(x):=\lim_{t\to 0}\frac{1}{t} \int_{y\in\mathcal{S}}\left( g(y)P(t,x,dy)-g(x)\right).
\end{align*}
We say that $g$ belongs to the domain of the generator $G$ of the process $X$ and write $g \in D(G)$ if the above limit exists for all $x\in \mathcal{S}$. The first piece of our novel BAR result is as follows. We prove this result in \Cref{apdx:sec7}.

\begin{restatable}{proposition}{propdrift}
\label{prop:drift}
Suppose that $g\in D(G)$ and $(X_t)_{t\geq 0}$ has a stationary distribution $\pi$, for which $|g|$ and $|G\circ g|$ are $\pi$-integrable. Moreover, suppose that the following holds for all $t$ and $x$:
\begin{align}
\label{eqn:uniform}
    \frac{1}{t} \int_\mathcal{S}\left( g(y)P(t,x,dy)-g(x)\right)=A(t,x)+B(t,x),
\end{align}
where $\lim_{t\to 0}A(t,\cdot)\to G\circ g(\cdot)$ uniformly and $\lim_{t\to 0}\int_S B(t,x)\,\pi(dx)=0$. Then the BAR holds:
\begin{align}
\int_\mathcal{S}G\circ g(x)\,\pi(dx)=0.
\end{align}
\end{restatable}

\subsection{Applying drift method to ISQ-$k$ for quadratic test functions}
\label{sec:drift_isq}
Next, we derive the necessary conditions to apply the drift method and specifically \Cref{prop:drift}, to the ISQ-$k$ system. 
Throughout this paper, we will always use test functions with a leading $w^2$ term. Therefore, we present our result for functions of the form $g(w,i)=w^2+c(w,i),$ where $c(w, i)$ is linear in $w$. It is important for our result that the function is at most quadratic in $w$, and that the largest term whose coefficient depends on $i$ is at most linear in $w$. This is the setting in which we prove our novel BAR results in \Cref{lem:uniform}.

The first lemma provides sufficient conditions for the finiteness of $\E[W^2]$, namely that the job size distribution has a finite third moment. This holds for all phase-type distributions, as well as for the truncated distributions that arise in our analysis of relevant work. We prove the following lemma in \Cref{apdx:sec7}.

\begin{restatable}{lemma}{lemwsquare}
\label{lem:w2}
In an ISQ-$k$ system, if the job size $S$ has a finite third moment, $E[S^3] < \infty$, then $\E[W^2]<\infty$.
\end{restatable}

Next, we provide sufficient conditions on the test function $g$ to ensure that $G\circ g$ satisfies \Cref{prop:drift}'s condition. This result makes use of the fact that the arrival process is Poisson and we prove it in \Cref{apdx:sec7}.

\begin{restatable}{lemma}{lemuniform}
\label{lem:uniform}
In an ISQ-$k$ system, suppose $\E[S^3]<\infty$ and $g(w,i)=w^2+c(w,i)$ is a real-valued function of the ISQ-$k$ system which is twice-differentiable with respect to $w$ for each fixed $i$.
Suppose that $|c(w,i)|\leq C_1w+C_2$ for some constants $C_1$ and $C_2$ and $\lim_{w\to0^{+}}g(w,i)=g(0,0)$ for each fixed $i$. Moreover, we assume $|c'(w,i)|\leq M_1$ and $|c''(w,i)|\leq M_2$ for some constants $M_1$ and $M_2$. Then \begin{align}
   \dfrac{1}{t}\E[g(W(t),I(t))-g(w,i)\,|\,W(0)=w, I(0)=i]=A(t,w,i)+B(t,w,i),
\end{align}
where $\lim_{t \to 0} A(t,\cdot,\cdot)\to G\circ g(\cdot,\cdot)$ uniformly and $\lim_{t\to 0}\E[B(t,W,I)]=0$.
\end{restatable}

Now, we are ready to prove our novel BAR result. Combining \Cref{lem:w2,lem:uniform} we have the following BAR \eqref{eqn:steady_state} result which shows that in the ISQ-$k$ system, for test functions $g$ of the structure described here, the expected value of the $G \circ g$ in steady state is zero. We prove the following lemma in \Cref{apdx:sec7}.

\begin{lemma}
\label{lem:drift}
In an ISQ-$k$ system, suppose $\E[S^3]<\infty$ and $g(w,i)=w^2+c(w,i)$ is a real-valued function of the ISQ-$k$ system which is twice-differentiable with respect to $w$ for each fixed $i$.
Suppose that $|c(w,i)|\leq C_1w+C_2$ for some constants $C_1$ and $C_2$ and $\lim_{w\to0^{+}}g(w,i)=g(0,0)$ for each fixed $i$. Moreover, we assume $|c'(w,i)|\leq M_1$ and $|c''(w,i)|\leq M_2$ for some constants $M_1$ and $M_2$. Then,
\begin{align}
    \E[G\circ g(W,I)]=0,
\end{align}
where the expectation is taken over the stationary random variables $W$ and $I$.
\end{lemma}

We will often apply \Cref{lem:drift} when the job size distribution follows a truncated distribution such as $S_x$ and $S_{\bar{x}}$, for which the assumption that $\E[S^3] <\infty$ is automatically satisfied. It is straightforward to apply \Cref{lem:drift} to the Sep-ISQ-$k$ system because each subsystem is Markovian and independent. For the Rec-ISQ-$k$ system, we only apply \Cref{lem:drift} to the Poisson arrival stream and deal with the recycling stream through a different approach (\citet{Braverman2017} and \citet{Braverman2024}).

In particular, we specify a version of \Cref{lem:isq_drift} for the Rec-ISQ-$k$ system using the Palm expectation $\E_\mathfrak{r}$ over the moments when jobs recycle, following \citet{Braverman2017} and \citet{Braverman2024}. Note that the state of the Rec-ISQ-$k$ system is $(w, i, r)$, where $w$ is the work, $i$ is the speed, and $r$ is the recycling arrival state. In this paper, we only consider test functions which do not depend on $r$. When it is clear, we write these test functions as $g(w, i)$. For clarity, we write test functions with three inputs $(w,i,r)$ in the following lemma:

\begin{lemma}
\label{lem:ar_drift}
    For any real-valued differentiable function $g$ of the state of the Rec-ISQ-$k$ system which does not depend on the recycling arrival state $r$,
    \begin{align*}
        G\circ g(w,i,r)=\lambda_x\E[g(w+S,\min\{1,i+1/k\},\cdot)-g(w,i,\cdot)]-\frac{d}{dw}g(w,i,\cdot)i + (\lambda - \lambda_x)\E_{\mathfrak{r} \mid w, i, r}[J(W,I)],
    \end{align*}
    where $J(w,i) =g(w+x,\min\{i+1/k, 1\},\cdot) - g(w,i,\cdot)$
    denotes the increase in the test function due to the arrival of a size-$x$ job,
    and $E_{\mathfrak{r} \mid w, i, r}[\cdot]$ denotes the conditional expectation of recycling with respect to the Palm measure in the immediate future of the state $(w, i, r)$ of the Rec-ISQ-$k$ system.
\end{lemma}

\section{Bounding Mean Relevant Work in the $M/G/2$}
\label{sec:mg2}

In this section, we derive explicit lower bounds on mean relevant work in the $M/G/2$ using the Sep-ISQ-2 and Rec-ISQ-2 systems. We start with the 2-server system because it demonstrates the core idea of our proof before we move on to the more general results concerning the $M/G/k$.
In \Cref{sec:empirical}, we numerically compare our novel bounds on mean relevant work to the existing bounds in the literature.

We will use the following test functions to bound mean relevant work in the Sep-ISQ-2 and Rec-ISQ-2 systems. We explain the intuition behind these test functions in \Cref{sec:derive}.
\begin{definition}
\label{def:isq2_constant}
We define the ISQ-2 constant-drift test function $g_2$ as follows,
\begin{align}
\label{eqn:isq2_constant}
g_2(w,1)=w,\quad g_2(0,0)=0\quad \text{and}\quad g_2(w,1/2)=w+\dfrac{1-e^{-2\lambda w}}{2\lambda}.
\end{align}
\end{definition}

\begin{definition}
\label{def:isq2_affine}
 We define the ISQ-2 affine-drift test function $h_2$ as follows,
\begin{align}
\label{eqn:isq2_affine}
    h_2(w,1)=w^2,\quad h_2(0,0)=0\quad\text{and}\quad h_2(w,1/2)=w^2+\dfrac{w}{\lambda}-\dfrac{1-e^{-2w\lambda}}{2\lambda^2}.
\end{align}
\end{definition}

It is easy to verify that $g_2$ and $h_2$ satisfy the assumptions of \Cref{lem:drift}, our BAR result. We now characterize the mean work of the ISQ-2 system by applying \Cref{lem:drift} to the test functions $g_2$ and $h_2$. 
\begin{proposition}
\label{prop:isq2_work}
    For any job size $S$ such that $E[S^3]$ is finite, and any arrival rate $\lambda$, the expected total work in the ISQ-$2$ system is given by 
    \begin{align}
    \label{eqn:isq2_work}
        \E[W^{\text{ISQ-2}}]=\dfrac{\lambda\E[S^2]}{2(1-\lambda\E[S])}+\dfrac{\E[S]-(1-\widetilde{S}(2\lambda))/2\lambda}{3-\widetilde{S}(2\lambda)}.
    \end{align}
\end{proposition}
\begin{proof}
    In order to characterize the mean work of the ISQ-$k$ system, we first need to characterize the fraction of the time that the system is idle, $\mathbb{P}(I=0)$.
    
    To do so, we first consider the constant-drift test function $g_2$,
    defined in
    \Cref{def:isq2_constant}.
    We want to calculate the drift $G \circ g_2(w, i)$
    for all possible values of $w$ and $i$.
    Recall from \Cref{sec:isq,def:isq2_constant} that
    the possible values of the speed $i$ are speeds $0, 1/2,$ and $1$.
    When $i=0$, the work must be $0$ by definition,
    while if $i>0$, the work must be positive.
    
    When $w>0$ and $i=1$, $G\circ g_2(w,i)=\lambda \E[S]-1.$ Arrivals cause a drift of $\lambda E[S]$,
    while work completion causes a drift of $-1$. When $w>0$ and $i=1/2$, applying \Cref{lem:isq_drift}, we have $G\circ g_2(w,i)=\lambda \E[S]-1.$
    The choice of the function $g_2$ ensures this drift property (See \Cref{sec:derive}). When $w=0$ and $i=0$, 
    \begin{align*}
    G\circ g_2(0,0)=\dfrac{1}{2}\left(1+2\lambda\E[S]-\widetilde{S}(2\lambda) \right)=\lambda \E[S]-1+\left(\frac{3}{2}-\dfrac{1}{2}\widetilde{S}(2\lambda)\right).
    \end{align*}
    Thus, we can summarize the drift over all states as
    \begin{align*}
    G\circ g_2(w,i)=\lambda \E[S]-1 + \left(\frac{3}{2}-\frac{1}{2}\widetilde{S}(2\lambda)\right)\mathbbm{1}_{\{i=0\}}.
    \end{align*}
    Setting the expectation to zero by \Cref{lem:drift}
    we have,
    \begin{align}
    \label{eqn:isq2_prob0}
    \mathbb{P}(I=0)=\dfrac{2(1-\lambda \E[S])}{3-\widetilde{S}(2\lambda)}.
    \end{align}
    Now we switch our focus to characterizing the mean work of the ISQ-$k$ system. We consider the affine-drift test function, $h_2$, defined in \Cref{def:isq2_affine}. When $w>0$ and $i=1$ or $1/2$,$ G\circ h_2(w,i)=\lambda\E[S^2]+2w(-1+\lambda\E[S]).$
    When $w=0$ we have, $G\circ h_2(0,0)=\lambda\E[S^2]+\E[S]+\frac{\widetilde{S}(2\lambda)-1}{2\lambda}.$
    Together, we can summarize the drift over all states as
    \begin{align*}
    G\circ h_2(w,i)=\lambda\E[S^2]+2w(-1+\lambda\E[S])+\left(\E[S]+\dfrac{\widetilde{S}(2\lambda)-1}{2\lambda}\right)\mathbbm{1}_{\{i=0\}}.
    \end{align*}
    Therefore, taking expectation of $G \circ h_2$ 
    and equating to zero using \Cref{lem:drift}, we have
    \begin{align*}
    \E[W^{ISQ{\text -}2}]&=\dfrac{\lambda\E[S^2]}{2(1-\lambda\E[S])}+\dfrac{\E[S]-(1-\widetilde{S}(2\lambda))/2\lambda}{2(1-\lambda\E[S])}\cdot\mathbb{P}(I=0)=\dfrac{\lambda\E[S^2]}{2(1-\lambda\E[S])}+\dfrac{\E[S]-(1-\widetilde{S}(2\lambda))/2\lambda}{3-\widetilde{S}(2\lambda)},
    \end{align*}
    where $\mathbb{P}(I=0)$ is given by \Cref{eqn:isq2_prob0}.

\end{proof}

Using \Cref{prop:isq2_work}, we can characterize the mean work of the truncated-ISQ-2 system in the Sep-ISQ-2 system. This provides us an exact mean relevant work formula for the Sep-ISQ-2 system. We therefore lower bound the mean relevant work of an $M/G/k$ system under arbitrary scheduling policy by \Cref{thm:ar_2}. This constitutes the first novel relevant work lower bound of \Cref{thm:main_2}, namely \Cref{eqn:main_2_1}. Our characterization of mean relevant work in the Sep-ISQ-2 system is as follows:

\begin{theorem}
\label{thm:sep_2}
    For an arbitrary threshold $x$, arbitrary job size $S$, and arbitrary arrival rate $\lambda$, the expected relevant work in the separate-ISQ-2 system is exactly given by
    \begin{align}
    \label{eqn:isq2_sep}
        \E[W_x^{\text{Sep-ISQ-}2}]=\dfrac{\lambda_x\E[S_x^2]}{2(1-\lambda_x\E[S_x])}+\dfrac{\E[S_x]-(1-\widetilde{S_x}(2\lambda_x))/2\lambda_x}{3-\widetilde{S_x}(2\lambda_x)}+\lambda(1-F_S(x))x^2.
    \end{align}
\end{theorem}

\begin{proof}
    The relevant work in the separate-ISQ-2 system is the sum of the total work in the truncated-ISQ-2 system and the relevant work in the separate $M/G/\infty$ system. Note that all work in the truncated-ISQ-2 system is relevant.
    
    First, let us handle the separate $M/G/\infty$ system.
    The servers at the $M/G/\infty$ system operate at a speed of $1/2$, and the arrival rate into this system is determined by $\lambda(1-F_S(x))$. Thus the expected relevant work at this separate server is $\lambda(1-F_S(x))x^2$. 

    Jobs arrive into the truncated-ISQ-2 with a conditional size $S_x$, having a density given by $f_S(x)/F_S(x)$ and bounded third moment $\E[S_x^3]\leq x^3$. The arrival rate is $\lambda_x := \lambda F_S(x)$. In particular, $\E[S_x]=\int_0^x sf(s)/F_S(x) ds$, $\E[S^2_x]=\int_0^x s^2f(s)/F_S(x) ds$, and $\widetilde{S}_x(2\lambda_x)=\int_0^x e^{-2\lambda_x s}f_S(s)/F_S(x) ds$. By \Cref{prop:isq2_work}, we have the following formula for expected relevant work formula for the truncated-ISQ-2 system:
    \begin{align*}
    \E[W_x^{{\text{truncated -ISQ-}}2}]=\dfrac{\lambda_x\E[S_x^2]}{2(1-\lambda_x\E[S_x])}+\dfrac{\E[S_x]-(1-\widetilde{S_x}(2\lambda_x))/2\lambda_x}{3-\widetilde{S_x}(2\lambda_x)}.
    \end{align*} 
    Adding the two expected relevant work formulas together we have \Cref{eqn:isq2_sep}.

\end{proof}

To derive the Rec-ISQ-2 lower bound, we will use the following modified ISQ-2 affine-drift test function $h_{2,x}$. Note that $h_{2,x}$ is a different function for each value of $x$. We provide intuition on how we derive $h_{2,x}$ in \Cref{sec:derive}. Again, it is easy to verify that the modified ISQ-2 affine-drift test function satisfies the assumptions of \Cref{lem:drift}.

\begin{definition}
\label{def:isq2_modified}
The ISQ-2 modified affine-drift test function $h_{2,x}$ is $h_{2,x}(0,0)=0$,
\begin{align}
\label{eqn:isq2_modified_affine}
h_{2,x}(w,1)=w^2\quad\text{and}\quad h_{2,x}(w,1/2)=w^2+\frac{w}{\lambda_x}-\frac{1-e^{-2w\lambda_x}}{2\lambda_x^2}-\frac{C_2(x,\lambda_x)(1-e^{-2w\lambda_x})}{\lambda_x},
\end{align} 
where $C_2(x,\lambda_x):=\frac{\E[S_x]-(1-\widetilde{S_x}(2\lambda_x))/2\lambda_x}{3-\widetilde{S_x}(2\lambda_x)}.$
Here, the first argument of $C_2$ is the relevancy cutoff level $x$ of the recycling ISQ-2 system.

\end{definition}

We are ready to state the main result for our Rec-ISQ-2 lower bound, which provides us with another lower bound in mean relevant work in the $M/G/k$ by \Cref{thm:ar_2}. This constitutes the second lower bound of \Cref{thm:main_2}, namely \Cref{eqn:main_2_2}.

\begin{restatable}{theorem}{thmarisq}
\label{thm:ar_2}
 For an arbitrary threshold $x$, job size $S$, arrival rate $\lambda$, and arbitrary recycling stream, expected relevant work $\E[W_{x}^{\text{Rec-ISQ-2}}]$ in the Rec-ISQ-2 system is lower bounded by
    \begin{align}
    \label{eqn:isq2_arb}
        \E[W_{x}^{\text{Rec-ISQ-2-}L}]:=\dfrac{\lambda_x \E[S_x^2]}{2 (1 - \rho_x)} + \dfrac{\E[S_x]-(1-\widetilde{S_x}(2\lambda_x))/2\lambda_x}{3-\widetilde{S_x}(2\lambda_x)}\cdot\dfrac{1-\rho_{\overbar{x}}}{1-\rho_x}+\dfrac{(\lambda - \lambda_x)x^2}{2(1-\rho_x)}.
    \end{align}
\end{restatable}

\begin{proof}
Let $\E_\mathfrak{r}[\cdot]$ denote the Palm expectation taken over the moments when the arbitrary arrival of jobs of size $x$ occurs according to the recycling stream. By \Cref{lem:isq_drift}, the expected drift $G \circ h_{2,x}$ has two kinds of terms, the drift due to Poisson arrivals of rate $\lambda_x = \lambda F_S(x)$ and the recycling jumps due to the arbitrary arrivals of jobs of size $x$ with rate $\lambda (1-F_S(x))$.

Let $J_x(w, i) := h_{2,x}(w+x,\min(i+1/2, 1)) - h_{2,x}(w,i)$
denote the increase in the test function due to the arrival of a size-$x$ job.
Then $\E_\mathfrak{r}[J(W_x,I)]$ denotes the mean size of the recycling jump in stationarity. The expectation $\E_\mathfrak{r}[\cdot]$ can be interpreted as the Palm expectation associated with the random measure that records the cumulative number of recyclings.
We define $\E_{\mathfrak{r} \mid w, i, r}[\cdot]$ to be the conditional Palm expectation in the immediate future of a specific state of the Rec-ISQ-$k$ system (\citet{Braverman2024}, \citet{Miyazawa1994}, and \citet{Scully2020}).
We start by applying \Cref{lem:ar_drift} to find the drift in a specific state:
\begin{align*}
    G \circ h_{2,x}(w, i) =\lambda_x\E[S_x^2]+2w(-1+\lambda\E[S_x])+2C_2(x,\lambda_x)\left(\mathbbm{1}_{\{i=0\}}+\frac{1}{2}\mathbbm{1}_{\{i=1/2\}} 
    \right) + (\lambda - \lambda_x) \E_{\mathfrak{r} \mid w, i, r}[J_x(w,i)]
\end{align*}

Applying our BAR result \Cref{lem:isq_drift} to the drift $G \circ h_{2,x}$, we find that
\begin{align*}
    0 = \E[G \circ h_{2,x}(W_x, I)] &=  \lambda_x\E[S_x^2]+2\E[W_x](-1+\lambda_x\E[S_x])+2C_2(x,\lambda)\left(\mathbb{P}(I^r=0)+\frac{1}{2}\mathbb{P}(I^r=1/2)\right)\\
    &\quad\quad + (\lambda - \lambda_x) \E_\mathfrak{r}[J_x(W_x,I)].
\end{align*}
Solving for $\E[W_x]$, we get
\begin{align}
\label{eqn:isq2_wx}
    \E[W_x] = \frac{\lambda_x E[S_x^2]}{2 (1 - \rho_x)} + \dfrac{C_2(x,\lambda_x)}{1-\rho_x}\left(\mathbb{P}(I^r=0)+\frac{1}{2}\mathbb{P}(I^r=1/2)\right) + \frac{(\lambda - \lambda_x) \E_\mathfrak{r}[J_x(W_x,I)]}{2(1-\rho_x)}.
\end{align}
In the above equation, we write $I^r$ instead of $I$ to make explicit the dependence of $I$ on recycling arrivals, and to avoid confusion with the probability $\mathbb{P}(I=0)$ defined in \Cref{eqn:isq2_prob0}.

Both $\mathbb{P}(I^r=0)$ and $\mathbb{P}(I^r=1/2)$ are difficult to evaluate. However, we show that we can relate the two probabilities to the capped load $\rho_{\overbar{x}}$.
To do so, we apply \Cref{lem:ar_drift} to the test function $g(w,i)=w$ for all $(w,i)$, we have
\begin{align*}
    0=\lambda_x\E[S_x]-1+\mathbb{P}(I^r=0)+\frac{1}{2}\mathbb{P}(I^r=1/2)+(\lambda-\lambda_x)x
\end{align*}
Re-arranging, and using the definition that $\rho_{\overbar{x}}=\lambda_x\E[S_x]+(\lambda-\lambda_x)x$, we get
\begin{align}
\label{eqn:isq2_load}
    \mathbb{P}(I^r=0)+\frac{1}{2}\mathbb{P}(I^r=1/2)=1-\rho_{\overbar{x}}.
\end{align}
Plugging in \Cref{eqn:isq2_load} into \Cref{eqn:isq2_wx}  we get
\begin{align*}
    \E[W_x] = \frac{\lambda_x E[S_x^2]}{2 (1 - \rho_x)} + \dfrac{C_2(x,\lambda_x)(1-\rho_{\overbar{x}})}{1-\rho_x} + \frac{(\lambda - \lambda_x) \E_\mathfrak{r}[J_x(W_x,I)]}{2(1-\rho_x)}.
    \end{align*}
    
Next, we want to lower bound $J_x(w,i)$ over all possible states in which a recycling could occur, and specifically the three cases $i=1,1/2,0$, respectively. In particular, we show that $J_x(w,i)$  is uniformly lower bounded by $x^2$ for any arbitrary job size $S$ and arrival rate $\lambda$. 
\begin{itemize}
    \item Jump size $J_x(w,i)$ at speed $i=1$ is given by $h_{2,x}(w+x,1)-h_{2,x}(w,1)=2wx+x^2\geq x^2$. 

    \item Jump size $J_x(w,i)$ at speed $i=1/2$ is given by
    \begin{align*}
        h_{2,x}(w+x, 1) - h_{2,x}(w, 1/2) &=2wx + x^2 - \frac{w}{\lambda_x} + \dfrac{1-e^{-2w\lambda_x}}{2\lambda_x^2}+\dfrac{C_2(x,\lambda_x)(1-e^{-2x\lambda_x})}{\lambda_x}\\
        &\geq 2wx + x^2 - \frac{w}{\lambda_x} + \dfrac{1-e^{-2w\lambda_x}}{2\lambda_x^2}.
    \end{align*}
    Let $J_{lb}(w,x) = 2wx + x^2 - \frac{w}{\lambda_x} + \frac{1-e^{-2w\lambda_x}}{2\lambda_x^2}$ denote this lower bound.
    
    We now minimize $J_{lb}$ over $w$.
    To do so, we split into two cases: $x > \frac{1}{2\lambda_x}$,
    and $x \le \frac{1}{2\lambda_x}$.
    If $x>\frac{1}{2\lambda_x}$, then $J_{lb}(w,1/2)$ is a concave increasing function in $w$. Thus, the function attains its minimum when $w=0$. Otherwise, $J_{lb}(w,1/2)$ is a concave function with a unique maximum. Because $J_{lb}$ is a concave function, its minimum must be either $w=0$ or $w=x$. The following calculation shows that the minimum is always at least $x^2$.
    \begin{align*}
        J_{lb}(x,x) &= 3x^2-\frac{x}{\lambda_x}-\frac{e^{-2x\lambda_x}-1}{2\lambda_x^2}\geq 3x^2-\frac{x}{\lambda_x}-\frac{1-2x\lambda_x+2x^2\lambda_x^2-1}{2\lambda_x^2}=2x^2\geq x^2 \\
        J_{lb}(0, x) &= x^2  \ge x^2.
    \end{align*}

    \item Jump $J_x(w,i)$ at speed $i=0$ is given by $x^2 + \frac{x}{\lambda_x} - \frac{1 - e^{-2x\lambda_x}}{2\lambda_x^2}-\frac{C_2(x,\lambda_x)(1-e^{-2x\lambda_x})}{\lambda_x}$.

    To lower bound $J_x(0, 0)$, we must upper bound $C_2(x, \lambda_x)$, which is defined in \Cref{def:isq2_modified}.
    Using the fact that $\E[S_x] \le x$ and that $\widetilde{S_x}(\cdot) \le 1$, we find that
    \begin{align*}
        C_2(x,\lambda_x)=\dfrac{\E[S_x]-(1-\widetilde{S_x}(2\lambda_x))/2\lambda_x}{3-\widetilde{S_x}(2\lambda_x)}\leq\dfrac{x-(1-1)/2\lambda_x}{3-1}=\dfrac{x}{2}.
    \end{align*}    
    
    Now, we want to lower bound $J_x(0, 0)$. Note that $C_2(x, \lambda_x)$ has a negative coefficient in the formula for $J_x(0, 0)$, allowing us to apply our upper bound on $C_2(x, \lambda_x)$ to derive a lower bound on $J_x(0, 0)$: 
    \begin{align*}
        J_x(0, 0) &= x^2 + \frac{x}{\lambda_x} - \frac{1 - e^{-2x\lambda_x}}{2\lambda_x^2}-\frac{C_2(x,\lambda_x)(1-e^{-2x\lambda_x})}{\lambda_x}
        \\&\ge x^2 + \frac{x}{\lambda_x} - \frac{1 - e^{-2x\lambda_x}}{2\lambda_x^2}-\frac{x(1-e^{-2x\lambda_x})}{2\lambda_x}.
    \end{align*}
    Therefore, it remains to show that the last three terms on the RHS of the above equation, which we define as the function $r(x)$, are nonnegative:
    \begin{align*}
     r(x):=\frac{x}{\lambda_x} - \frac{1 - e^{-2x\lambda_x}}{2\lambda_x^2}-\frac{x(1-e^{-2x\lambda_x})}{2\lambda_x}=\frac{2x\lambda_x-(1-e^{-2x\lambda_x})(1+x\lambda_x)}{2\lambda_x^2}.
    \end{align*}
    If we show that $r(x) \ge 0$ for all $x \ge 0$, then we have shown that $J(0,0)$ is lower bounded by $x^2$. 
    
    This is equivalent to showing that the numerator of the last term above is nonnegative, i.e.,
    \begin{align*}
        2x\lambda_x-(1-e^{-2x\lambda_x})(1+x\lambda_x)=(x\lambda_x-1)+(1+x\lambda_x)e^{-2x\lambda_x}\geq 0.
    \end{align*}
    Let $y:=x\lambda_x$. Note that $y>0$. Then the inequality becomes
    $(y-1)+(1+y)e^{-2y}\geq 0 \iff (1+y)e^{-2y}\geq 1-y$ which is trivially true if $y\geq 1$. For $y\in(0,1)$, we show that the equivalent inequality $\frac{1+y}{1-y}\geq e^{2y}$ holds.
    Because $\ln(\cdot)$ is increasing and $\frac{1+y}{1-y}$ is positive for all $y \in (0, 1)$, we equivalently show that $\ln \frac{1+y}{1-y} \geq 2y$.
    
    We start with the Taylor expansions $\ln(1+y)=y-y^2/2+y^3/3-\cdots$ and $\ln(1-y)=-y-y^2/2-y^3/3-\cdots$. We therefore have, $\ln\left(\frac{1+y}{1-y}\right)=\ln(1+y)-\ln(1-y)=2y+\frac{2y^3}{3}+\frac{2y^5}{5}+\cdots\geq 2y.$
    Exponentiating both sides, we have $\frac{1+y}{1-y}\geq e^{2y}$ for $y\in(0,1)$, so $r(x) \ge 0$ for all $x \ge 0$. Thus, $J(0, 0) \ge x^2$.

\end{itemize}    

Therefore, because $J_x(w,i)$ is lower bounded by $x^2$ for all $w,i$,
we know that $\E_\mathfrak{r}[J_x(W_x, I)]$ is also lower bounded by $x^2$.
We therefore have our desired lower bound on the mean relevant work in the Rec-ISQ-2 system:
    \begin{align*}
        \E[W_{x}^{\text{Rec-ISQ-2}}]\ge\frac{\lambda_x \E[S_x^2]}{2 (1 - \rho_x)} + \dfrac{C_2(x,\lambda)(1-\rho_{\overbar{x}})}{1-\rho_x}+\frac{(\lambda - \lambda_x)x^2}{2(1-\rho_x)}:=\E[W_{x}^{\text{Rec-ISQ-2-}L}]. 
    \end{align*}
\end{proof}

\subsection{Proof of \Cref{thm:main_2}}
\label{sec:thm_main_2}
In this subsection, we prove \Cref{thm:main_2}.
\begin{proof}
To prove \Cref{thm:main_2} we combine results established in \Cref{sec:wine} and \Cref{sec:mg2}.  
We obtain \Cref{eqn:main_mixex2} by applying the WINE formula to \Cref{eqn:lower_1,eqn:lower_2}.

\Cref{eqn:main_isqs2} further incorporates $\E[W_x^{\text{Sep-ISQ-2}}]$. The fact that Sep-ISQ-2 lower bounds the relevant work of any $M/G/2$ system under arbitrary scheduling policy is proved in \Cref{thm:isq_sep} and we derive its analytical expression given by \Cref{eqn:main_2_1} in \Cref{thm:sep_2}.

Finally, \Cref{eqn:main_isqr2} incorporates $\E[W_x^{\text{Rec-ISQ-2-}L}]$. The fact that there exists a Rec-ISQ-2 system lower bounding the relevant work of  
any $M/G/2$ system under arbitrary scheduling policy is proved in \Cref{thm:ar_isq} and we derive a lower bound on $\E[W_x^{\text{Rec-ISQ-2}}]$ given by \Cref{eqn:main_2_2} in \Cref{thm:sep_2}.

\end{proof}

\section{Deriving the Test Functions -- DiffeDrift}
\label{sec:derive}
In this section, we present our DiffeDrift method, which builds upon the drift method/BAR approach from prior literature. In DiffeDrift, we first select the desired drift and then derive the corresponding test function using differential equations. We focus on the ISQ-2 test functions introduced in \Cref{sec:mg2}, as they illustrate the main concept of the method. We generalize these test functions to the general case in \Cref{sec:mgk}.

We start by deriving the affine-drift test function (\Cref{def:isq2_affine}) in \Cref{sec:derive_affine} as well as the constant-drift test function (\Cref{def:isq2_constant}), which we used to characterize the mean total work of an ISQ-2 system. We then derive the modified affine-drift test function (\Cref{def:isq2_modified}) in \Cref{sec:derive_modified}, which is a specialized test function for the Rec-ISQ-2 system.

\subsection{Affine-drift and constant-drift test functions} 
\label{sec:derive_affine}
Recall that the instantaneous drift operator $G$ is a stochastic version of a derivative operator. To find information about the mean work in the system, we consider test functions with a leading quadratic term $w^2$. The drift of such a test function is a linear function of $w$ and by applying \Cref{lem:drift} we plan to solve for $\E[W]$.

The possible states of the ISQ-2 system are $(0,0),(w,1/2)$ and $(w,1)$. There are three events in the system that can affect the state: stochastic arrivals, deterministic decrease in $w$ and the completion of the final job in the system. 

Due to arrivals, $w$ increases as stochastic jumps of size $S$ arrive at rate $\lambda$. When $i>0$, due to work completion $w$ decreases at rate $i$. Therefore, using \Cref{lem:isq_drift} we can write down the drift for any arbitrary test function $h$,
\begin{align*}
    G\circ h(w,i)=\lambda(\E[h(w+S,\min\{i+1/k,1\})-h(w,i)])-h'(w,i)\cdot i.
\end{align*}
Note that for $h$ to satisfy the conditions of \Cref{lem:drift}, it must change continuously if no stochastic events occur. In particular, when a busy period ends it must be that $\lim_{w\to0^+}h(w,i)=h(0,0):=0$ for all $i$.

To derive the affine-drift test function $h_2$, we start with a simple expression for $h_2(0, 0)$ and $h_2(w, 1)$, and solve for the necessary form of $h_2(w, 1/2)$. We define $h_2(w,0)=0$ as above, and let $h_2(w,1)=w^2$.

Now, our goal is to define $h_2(w, 1/2)$ to ensure that the drift at speed $1/2$ matches the drift at speed $1$. This allows us to isolate the complexity of the drift function to the case $i=0$ when applying \Cref{lem:drift}.

We calculate the drift at speed $1$,
\begin{align*}
    G\circ h_2(w,1)=\lambda(\E[h_2(w+S,1)]-\E[h_2(w+S,1)])-h'_2(w,1)=\lambda\E[S^2]+2\lambda w\E[S]-2w.
\end{align*}

We can also write down the drift at speed $1/2$,
\begin{align*}
    G\circ h_2(w,1/2) &= \lambda (\E[h_2(w+S,1)-h_2(w,1/2)])-h'_2(w,1/2)\cdot \frac{1}{2}\\
    &=\lambda\E[S^2]+\lambda w^2+2\lambda w\E[S]-\lambda h_2(w,1/2) -\frac{h'_2(w,1/2)}{2}.
\end{align*}

By comparing $G\circ h_2(w,1)$ with $G\circ h_2(w,1/2)$, we see that the two drifts match if and only if $h_2(w,1/2)$ solves the following differential equation,
\begin{align*}
    \lambda w^2+2w-\lambda h_2(w,1/2)-\frac{h_2'(w,1/2)}{2}=0,\quad h_2(0,1/2)=0,
\end{align*}
which has a unique solution given by $h_2(w,1/2)=w^2-\frac{1-e^{-2w\lambda}}{2\lambda^2}+\frac{w}{\lambda}.$ Solving this differential equation is the essence of our DiffeDrift method. This defines the affine-drift test function $h_2(w,i)$ over the three possible states $(0,0), (w,1/2)$ and $(w,1)$, as given in \Cref{def:isq2_affine}. 

In the proof of \Cref{prop:isq2_work}, we see that the affine-drift test function $h_2$ alone is insufficient to derive  $\E[W]$, we also need to determine $\mathbb{P}(I=0)$. Since $w$ has a constant drift, we can characterize $\mathbb{P}(I=0)$ using a test function with a leading linear term in $w$. By following similar differential-equation-based steps, one can derive the constant-drift test function $g_2$ defined in \Cref{def:isq2_constant}. We generalize the affine-drift test function to the setting of a general number of servers $k$ in \Cref{lem:isqk_affine}.

\subsection{Modified affine-drift test function}
\label{sec:derive_modified}
The affine-drift and constant-drift test functions are sufficient to determine the mean work in the ISQ-2 system. They also suffice to determine the mean relevant work in the Sep-ISQ-2 system. However, in the Rec-ISQ-2 system an recycling stream with jobs of size $x$ arrives into the same system as the truncated stream. Therefore, $\mathbb{P}(I
=0)$ is no longer given by the expression derived in the proof of \Cref{prop:isq2_work} and we cannot apply the result of \Cref{prop:isq2_work} in determining the changes in the drift of the Rec-ISQ-2 system due to the truncated stream. Letting $I^r$ denote the speed distribution in $\mathbb{P}(I^r=0)$, which now depends on the specific recycling stream $\{R_t\}$, which is difficult to characterize exactly.

However, note that the load of the system in equilibrium does not depend on the recycling stream. Using the test function $g(w,i)=w$ for all $(w,i)$, we find that
\begin{align}
    \label{eq:unused-load}
    1-\rho_{\bar{x}}=\mathbb{P}(I^r=0)+\frac{1}{2}\mathbb{P}(I^r=1/2).
\end{align}
This is a characterization of the unused capacity in the Rec-ISQ-2 system, which is unaffected by the details of recycling stream. Therefore, our plan is to modify the affine-drift test function $h_2$ so that we get an unused-load term matching \eqref{eq:unused-load}. Additionally, we use \Cref{lem:ar_drift} instead of \Cref{lem:drift} because it supports recyclings. We start with the following test functions:
\begin{align*}
h_{2,x}(w,1)=w^2,\quad h_{2,x}(0,0)=0\quad\text{and}\quad h_{2,x}(w,1/2)=w^2+\ell_1(w),
\end{align*} 
for a function $\ell_1$ to be determined.

We do not want to match the drift at speed $1/2$ with the drift at speed $1$ which would results in a $\mathbb{P}(I^r=0)$ term. Instead, to obtain an unused-load term we want to choose $\ell_1(w)$ to ensure that
\begin{align}
    \label{eq:drift-derivation}
    G\circ h_{2,x}(w,1/2) - G\circ h_{2,x}(w,1) =\frac{1}{2}\left( G\circ h_{2,x}(0,0) - (\lim_{w\to 0^+} G\circ h_{2,x}(w,1))\right) =: C_2(x, \lambda_x).
\end{align}
In particular, we want to ensure that \eqref{eq:drift-derivation} is a constant not depending on $w$. By doing so, we will create a term matching the unused-load \eqref{eq:unused-load}, allowing us to effectively bound the Rec-ISQ-2 system. We call this constant $C_2(x, \lambda_x)$, though we do not yet know its exact value. The drift at speed $1/2$ is given by,
\begin{align*}
    G\circ h_{2,x}(w,1/2)&=\lambda_x(E[h_{2,x}(w+S_x,1)-h_{2,x}(w,1/2)])-\frac{1}{2}(h_{2,x}(w,1/2))'\\
    &=\lambda_x E[S_x^2]+2w(\lambda_x E[S_x]-1)+w-\lambda_x \ell_1(w)-\frac{\ell_1'(w)}{2}.
\end{align*}

Note that $G\circ h_{2,x}(w,1/2) - G\circ h_{2,x}(w,1) = w-\lambda_x \ell_1(w)-\frac{\ell_1'(w)}{2}$.
Our goal is to choose $\ell_1(w)$ to ensure that this quantity is a constant, which moreover matches $\frac{1}{2}\left( G\circ h_{2,x}(0,0) - G\circ h_{2,x}(0,1)\right)$. This is the heart of the DiffeDrift method.  

Solving the  differential equation $w-\lambda_x \ell_1(w)-\frac{\ell_1'(w)}{2}=C_2(x,\lambda_x)$ with $\ell_1(0)=0$, we find that,
\begin{align*}
    \ell_1(w)=\frac{w}{\lambda_x}-\frac{1-e^{-2w\lambda_x}}{2\lambda_x^2}-\frac{C_2(x,\lambda_x)(1-e^{-2w\lambda_x})}{\lambda_x}.
\end{align*}

Now, we can solve for $C_2(x, \lambda)$ as follows: By definition, we have $G\circ h_{2,x}(0,0)=\lambda_x \E[S_x^2]+\lambda_x \E[\ell_1(S_x)]$ and $\lim_{w \to 0^+} G\circ h_{2,x}(w,1) = \lambda_x E[S_x^2]$. Therefore,
\begin{align*}
    2C_2(x, \lambda_x) = \lambda_x \E[\ell_1(S_x)] \implies
    C_2(x, \lambda_x) = \dfrac{E[S_x]-(1-\widetilde{S_x}(2\lambda_x))/2\lambda_x}{3-\widetilde{S_x}(2\lambda_x)}.
\end{align*}

Now, we can exactly derive the expected value of the drift due to the truncated stream:
\begin{align*}
    \E[G\circ h_{2,x}(W_x,I)]&=\lambda_x\E[S_x^2]+2\E[W_x](-1+\lambda_x\E[S_x])+2C_2(x,\lambda_x)\left(\mathbb{P}(I^r=0)+\frac{1}{2}\mathbb{P}(I^r=1/2)\right)\\
    &=\lambda_x\E[S_x^2]+2\E[W_x](-1+\lambda_x\E[S_x])+2C_2(x,\lambda_x)(1-\rho_{\bar{x}}).
\end{align*}
We were able to apply the unused capacity formula \eqref{eq:unused-load} because our test function satisfied \eqref{eq:drift-derivation}. This is key to the strength of our Rec-ISQ-2 bound \Cref{thm:ar_2}.

We generalize the affine-drift test function to the setting of a general number of servers $k$ in \Cref{lem:isqk_modified}, using the same differential-equation structure for the DiffeDrift method.

\section{Bounding Mean Relevant Work in the $M/G/k$}
\label{sec:mgk}

In this section, we extend the results in \Cref{sec:mg2} for the ISQ-2 system to the general ISQ-$k$ system. The derivation of the test functions follows, the same ideas outlined in \Cref{sec:derive}. We start by deriving our test functions using the DiffeDrift method. We first define the ISQ-$k$ constant-drift test function.

\begin{definition}
\label{def:isqk_contant}
We define the ISQ-$k$ constant-drift test function $g_k$ as follows:
\begin{align*}
    g_k(w,0/k)=0,\,g_k(w,1/k)=w+u_1(w),\,\cdots,\, g_k(w,\frac{k-1}{k})=w+u_{k-1}(w),\, g_k(w,k/k)=w,
\end{align*}
where for each $q\in\{k-1,\cdots,1\}$, we define $u_q(w)$ by the following recursive formula,
\begin{align}
\label{eqn:isqk_constant}
    u_{q}(w)=e^{-\frac{kw\lambda}{q}}\int_0^w\dfrac{e^{\frac{k\lambda y}{q}}(k-q+k\lambda \E[u_{q+1}(S+y)])}{q}\,dy,
\end{align}
where $u_k(w)=0$.
\end{definition}

Note that for $k=2$ this simplifies to the $g_2$ expression defined in \Cref{def:isq2_constant}. Next, we define the ISQ-$k$ affine-drift test function.

\begin{definition}
\label{def:isqk_affine}
We define the ISQ-$k$ affine-drift test function as follows,
\begin{align*}
h_k(w,0/k)=0,\,h_k(w,1/k)=w^2+v_1(w),\,\cdots,\, h_k(w,\frac{k-1}{k})=w^2+v_{k-1}(w),\, h_k(w,k/k)=w^2.
\end{align*}
For each $q\in\{k-1,\cdots,1\}$, we define $v_q(w)$ by the following recursive formula,
\begin{align}
\label{eqn:isqk_affine}
    v_q(w)=e^{-\frac{kw\lambda}{q}}\int_0^w \dfrac{e^{\frac{k\lambda y}{q}}(2ky-2qy+k\lambda \E[v_{q+1}(S+y)])}{q}\,dy,
\end{align}
where $v_k(w)=0$. 
\end{definition}

Note that for $k=2$ this simplifies to the $h_2$ expression defined in \Cref{def:isq2_affine}. Note also that the $v_q(w)$ formulas have a similar recursive structure as $u_q(w)$, but with an additional $y$ coefficient in the numerator of the integrand fraction. \Cref{eqn:isqk_constant,eqn:isqk_affine} can be explicitly solved for any $k$, see \Cref{apdx:isq3} for the 3,4 and 5 server cases, and \Cref{sec:generalization} for further discussion on $k>5$. In \Cref{lem:isqk_constant_assumptions,lem:isqk_affined_bound} of \Cref{apdx:sec9} we prove the validity of these test functions.

Proceeding similarly as in the proof of \Cref{prop:isq2_work}, we compute the mean work of the ISQ-$k$. 

\begin{proposition}
\label{prop:isqk_work}
For an arbitrary job size $S$ such that $E[S^3]$ is finite and arrival rate $\lambda$, the expected total work in an ISQ-$k$ system is given by 
    \begin{align}
    \label{eqn:isqk_work}
        \E[W^{\text{ISQ-}k}]=\dfrac{\lambda\E[S^2]}{2(1-\lambda\E[S])}+\dfrac{\lambda\E[v_1(S)]}{2+2\lambda\E[u_1(S)]}.
    \end{align}    
\end{proposition}
\begin{proof}
By \Cref{lem:isqk_constant} and \Cref{lem:isqk_affine} in \Cref{apdx:sec9} the drifts of the constant and affine test functions at speed $0$ are given by $G\circ g_k(0,0)=\lambda\E[S+u_1(S)]$ and $G\circ h_k(0,0)=\lambda\E[S^2+v_1(S)]$ respectively. At all other speeds $i \geq 1/k$, the drift of the constant-drift test function is given by $G\circ g_k(w,i)=\lambda\E[S]-1$ and the drift of the affine-drift test function is given by $G\circ h_k(w,i)=\lambda\E[S^2]+2w(\lambda\E[S]-1)$.

The rest of the proof can be completed in a similar fashion as in the proof of \Cref{prop:isq2_work}. Summarizing the drift of $g_k$ over all states, we get, $G\circ g_k(w,i)=\lambda\E[S]-1+(1+\lambda\E[u_1(S)])\cdot\mathbbm{1}_{\{i=0\}}.$
By \Cref{lem:isqk_constant} the constant-drift test function $g_k$ satisfies the assumption of \Cref{lem:drift}. Thus, we have $\E[G\circ g_k(W,I)]=0$ and solving for $\E[\mathbbm{1}_{\{i=0\}}]=\mathbb{P}(I=0)$ we get the probability that the system is in speed 0,
\begin{align}
\label{eqn:isqk_prob0}
    \mathbb{P}(I=0)=\dfrac{1-\lambda\E[S]}{1+\lambda\E[u_1(S)]}.
\end{align}
Similarly, we can summarize the drift of $h_k$ as
\begin{align*}
    G\circ h_k(w,i)=\lambda\E[S^2]+2w(\lambda\E[S]-1)+\lambda\E[v_1(S)]\cdot\mathbbm{1}_{\{i=0\}}.
\end{align*}
By \Cref{lem:isqk_affine} in \Cref{apdx:sec9} the affine-drift test function $h_k$ satisfies the assumption of \Cref{lem:drift}. Thus, we have $\E[G\circ h_k(W,I)]=0$. Solving for $\E[W]$ we get,
\begin{align*}
    \E[W^{ISQ{\text -}k}]&=\dfrac{\lambda\E[S^2]}{2(1-\lambda\E[S])}+\dfrac{\lambda\E[v_1(S)]}{\lambda\E[S]}\cdot\mathbb{P}(I=0)=\dfrac{\lambda\E[S^2]}{2(1-\lambda\E[S])}+\dfrac{\lambda\E[v_1(S)]}{2(1+\lambda\E[u_1(S)])}
\end{align*}    
where $\mathbb{P}(I=0)$ is given by \Cref{eqn:isqk_prob0}.

\end{proof}

Next, using \Cref{prop:isqk_work} we give the exact mean relevant work in the Sep-ISQ-$k$ system. This constitutes the first novel relevant work lower bound of \Cref{thm:main_k}.

\begin{theorem}
\label{thm:sep_k}
For an arbitrary job size $S$ and arrival rate $\lambda$, the expected relevant work in the separate-ISQ-k system is given by,    
\begin{align}
    \label{eqn:isqk_sep}
    \E[W_x^{\text{Sep-ISQ-}k}]=\dfrac{\lambda_x\E[S_x^2]}{2(1-\lambda_x\E[S_x])}+\dfrac{\lambda_x\E[v_1(S_x)]}{2+2\lambda_x\E[u_1(S_x)]} +\dfrac{k\lambda(1-F_S(x))x^2}{2},
\end{align}
where $v_1$ and $u_1$ are defined in \Cref{eqn:isqk_constant,eqn:isqk_affine}.
\end{theorem}
\begin{proof}
The relevant work in the separate-ISQ-$k$ system is the sum of the total work in the truncated-ISQ-$k$ and the relevant work in the separate $M/G/\infty$ system.

The derivation for the relevant work in the truncated-ISQ-$k$ system is exactly the same as in the proof of \Cref{thm:sep_2} which makes up the first two terms on the RHS of \Cref{eqn:isqk_sep}.

For the separate $M/G/\infty$ system, the servers will now operate at a speed of $1/k$, and the arrival rate into this system is $\lambda(1-F_S(x))$. Thus the mean relevant work at this separate server is $\frac{k\lambda(1-F_S(x))x^2}{2}$. 
\end{proof}

For the last part of this section, we derive the Rec-ISQ-$k$ lower bound. We start with the ISQ-$k$ modified affine-test function, denoted by $h_{k,x}$.

\begin{definition}
\label{def:isqk_modified}
    We define the ISQ-$k$ modified affine-drift test function $h_{k,x}$ as follows:
    \begin{align*}
    h_{k,x}(w,0/k)=0,\,h_{k,x}(w,1/k)=w^2+\ell_1(w),\,\cdots,\, h_{k,x}(w,\frac{k-1}{k})=w^2+\ell_{k-1}(w),\, h_{k,x}(w,k/k)=w^2.
    \end{align*}
    For each $q\in\{k-1,\cdots,1\}$, we define $\ell_q(w)$ in terms of $\ell_{q+1}(w)$ by the following recursive formula,
    \begin{align}
    \label{eqn:isqk_modified}
        \ell_q(w)=e^{-\frac{kw\lambda_x}{q}}\int_0^w\dfrac{e^{\frac{k\lambda_x y}{q}}\left(k(q-k)C_k(x,\lambda_x)+2(k-q)y+k\lambda_x\E[\ell_{q+1}(S_x+y)] \right)}{q}\,dy
    \end{align}
    where $\ell_k(w)=0$ and 
    \begin{align}
    \label{eqn:C_k}
         C_k(x,\lambda_x)=\frac{2}{k}\cdot\dfrac{\lambda_x\E[v_1(S_x)]}{2+2\lambda_x\E[u_1(S_x)]}\geq 0,
    \end{align}
    where $v_1$ and $u_1$ are defined by \Cref{eqn:isqk_constant,eqn:isqk_affine}.
\end{definition}

The next theorem constitutes the second lower bound of \Cref{thm:main_k} and is based on the modified affine-drift test function. We prove this theorem in \Cref{apdx:sec9}.

\begin{restatable}{theorem}{mainarisqk}
\label{thm:ar_k}
For an arbitrary job size $S$ and arrival rate $\lambda$, a lower bound on expected relevant work $\E[W_x^{\text{Rec-ISQ-}k}]$ in the Rec-ISQ-$k$ system is given by,
    \begin{align}
    \label{eqn:isqk_jump}
    \E[W_x^{\text{Rec-ISQ-}k\text{-}L}]=
    \dfrac{\lambda_x\E[S_x^2]}{2(1-\lambda_x\E[S_x])}+\dfrac{\lambda_x\E[v_1(S_x)]}{2+2\lambda_x\E[u_1(S_x)]}\cdot\dfrac{1-\rho_{\bar{x}}}{1-\rho_x}+\frac{(\lambda - \lambda_x)J_x}{2(1-\rho_x)}.
\end{align}
Here $J_x$ is the smallest jump size incurred by the recycling stream, i.e., 
\begin{align*}
    J_x:=\min\left\{x^2, \min_{i < 1}\left\{
        \inf_{w\in[0,kix]}h_{k,x}(w+x,i+1/k)-h_{k,x}(w,i)
        \right\}\right\}.
\end{align*}
\end{restatable}

We believe that $J_x=x^2$ for all thresholds $x$, numbers of servers $k$, and job size $S$.
We proved in \Cref{thm:ar_2} that $J_x=x^2$ whenever $k=2$.
However, formal proof of this assertion for $k \ge 3$ remains elusive, so we leave this as an open problem. We have empirically verified that this is true for the 3-server setting with exponential job sizes.

For theoretical evidence, one can also observe that each term in $J_x$ is of the form
\begin{align*}
    x^2 + 2wx + \E[l_{k}(x+S)]-l_{k-1}(x),
\end{align*}
and the contribution of $x^2+2wx$ should dominate the residual $\E[l_{k}(x+S)]-l_{k-1}(x)$ across all $x\geq 0$, because the residual grows at most linearly, by \Cref{lem:isqk_modified_bound} in \Cref{apdx:sec9}. Therefore we propose the following conjecture:
\begin{conjecture}
\label{conj:x2}
    For any job size $S$ and any number of servers $k \ge 2, J_x=x^2$.
\end{conjecture}

\subsection{Proof of \Cref{thm:main_k}}
\label{sec:thm_main_k}
In this subsection, we prove \Cref{thm:main_k}.
\begin{proof}
To prove \Cref{thm:main_k} we combine results established in \Cref{sec:wine} and \Cref{sec:mgk}.  
We obtain \Cref{eqn:main_mixex_k} by applying the WINE formula to \Cref{eqn:lower_1,eqn:lower_2}.

\Cref{eqn:main_isqsk} further incorporates $\E[W_x^{\text{Sep-ISQ-}k}]$. The fact that Sep-ISQ-$k$ lower bounds the relevant work of any $M/G/k$ system under arbitrary scheduling policy is proved in \Cref{thm:isq_sep} and we derive its analytical expression given by \Cref{eqn:main_k_1} in \Cref{thm:sep_k}.

Finally, \Cref{eqn:main_isqrk} incorporates $\E[W_x^{\text{Rec-ISQ-}k}]$. The fact that there exists a Rec-ISQ-$k$ system lower bounding the relevant work of  
any $M/G/k$ system under arbitrary scheduling policy is proved in \Cref{thm:ar_isq} and we derive a lower bound on $\E[W_x^{\text{Rec-ISQ-k}}]$ given by \Cref{eqn:main_k_2} in \Cref{thm:sep_k}.

\end{proof}

\section{Discussion on Generalization}
\label{sec:generalization}

In this section, we discuss how our theoretical framework naturally extends to the setting where arrivals have job sizes that are unknown but drawn from a known distribution, or where arrivals come with estimates. 
In addition, we address how to solve the recursive formulas \Cref{eqn:isqk_constant,eqn:isqk_affine} for deriving ISQ-based lower bounds when $k$ is large.

\subsection{Unknown size}
\label{sec:gen_unknown}
Our framework can be extended to the setting where arrivals have job sizes that are unknown but drawn from a known distribution. 
The WINE formula, originally introduced by \citet{Scully2020}, shows that the mean response time can be expressed as
\begin{align}
\label{eqn:wine_2}
    \E[T] = \frac{1}{\lambda}\int_0^{\infty}\frac{\E[W_r]}{r^2}\,dr,
\end{align}
where $W_r$ denotes the total $r$-relevant Gittins work in the system. 
A job is said to be \emph{$r$-relevant} if its Gittins rank is less than the cutoff $r$, and $W_r$ is computed by summing over all jobs, the service each requires before its Gittins rank exceeds $r$. Note that this $r$ is unrelated to the recycling state, and we are using standard notation.

The \emph{Gittins rank} is a function of attained service $s$ and is defined as
\begin{align}
\label{eqn:gittins}
    \text{Gittins-rank}(s) = \inf_{t > s} \frac{\E[\min\{S,t\} - s \mid S > s]}{\mathbb{P}(S \leq t \mid S > s)}.
\end{align}
The Gittins rank was originally introduced by \citet{Gittins1979} to define the Gittins policy. The Gittins policy in scheduling selects the job with the smallest Gittins rank and is known to be optimal in the $M/G/1$ system.  However, the WINE formula \eqref{eqn:wine_2} is completely general, in particular, it holds for \emph{any} scheduling policy, even those unrelated to the Gittins policy.  

In the known size setting, the Gittins rank reduces to the remaining size, so $W_r$ is simply the sum of remaining work less than $r$, and \Cref{eqn:wine_2} reduces to \Cref{eqn:wine}.

The MixEx bound carries over naturally into this setting.  By computing the mean relevant work $\E[W_r]$ in the $M/G/\infty$ queue, together with the resource-pooled Gittins policy,  
we can plug these into the WINE framework to obtain the MixEx bound in the unknown job size setting, analogous to \Cref{eqn:main_mixex_k}. In addition, the Gittins-$k$ policy is the corresponding upper bound, which is heavy-traffic optimal in the unknown job size setting (\citet{Scully2020}).

To illustrate the Gittins policy and Gittins rank, let us consider two specific cases where the job size distribution follows either a decreasing hazard rate (DHR) or an increasing hazard rate (IHR). In the DHR case, all jobs arrive to the system with their minimum Gittins rank, which is equal to the reciprocal of their hazard rate. Therefore, the Gittins policy reduces to the least-attained-service (LAS) policy and the Gittins work is easy to compute. In the IHR case, jobs arrive to the system with rank $>0$ and their rank gradually decreases over time until completion. In this case, the Gittins policy reduces to FCFS. 

Next, we generalize the ISQ and ISQ-Recycling bounds in these two cases. In the DHR case, all jobs arrive to the system with their minimum rank and gradually increase over time until completion. Thus, there is no recycling in this system and our lower bound is given by applying the ISQ-$k$ formula to a truncated job size distribution.


In the IHR case, depending on the cutoff $r$ and the job size distribution, either all jobs enter the ISQ-$k$ system or all jobs recycle. Thus, ISQ only gives a nontrivial bound for ranks above the initial rank of the distribution and once again that lower bound is given by the ISQ-$k$ total work formula, \Cref{eqn:isqk_work}.

The Gittins rank function and the corresponding ISQ bounds are very simple in the IHR and DHR cases. For more complicated distributions, the Gittins rank function is non-trivial but related ISQ bounds can still be proven with additional effort.

\subsection{Estimated size}
\label{sec:gen_estimated}

Our framework can also incorporate the setting where arrivals come with estimates of job sizes. Just as in the case of unknown sizes, there are corresponding variants of the Gittins rank and corresponding WINE results in the setting where the joint distribution of size and estimate is known. For any such joint distribution, similar results to those
outlined in \Cref{sec:gen_unknown} are possible. 

When the estimates are of high quality, for jobs whose time in service is less than their estimate, the Gittins rank is approximately equal to the remaining size. Therefore, the analysis becomes more similar to the setting considered in this paper and the Gittins rank and relevant Gittins work are easier to compute, and thus similar ISQ bounds may be proven.

\subsection{Computational considerations for large $k$.}

In \Cref{apdx:isq3}, we explicitly derive the ISQ-$k$ relevant work lower bounds and response time lower bounds for $k=3,4,$ and $5$. We observe that the number of terms in these relevant work bounds roughly doubles as $k$ increases. Looking more carefully, the number of terms in $u_i$ and $v_i$, which are also given in \Cref{apdx:isq3}  double per iteration as $i$ decreases. As a result, manually writing down and computing the numerical integrals for $k>5$ becomes difficult and a standard symbolic algebra package such as \texttt{Mathematica} cease to function effectively. However, one could implement a specialized computational program to compute these recursive functions, as defined in \Cref{eqn:isqk_constant,eqn:isqk_affine}.

In particular, the recursive equations involve only three types of terms: one constant term, one linear term in $w$, and many structured exponential terms of the form
\begin{align}
z_{c,r,m,[q]}(w) = c \cdot \frac{e^{-r w \lambda \prod_i \widetilde{S}(q_i \lambda)}}{\lambda^m},
\end{align}
which can be characterized by a quadruple $(c, r, m, [q])$, where $c$ and $r$ are rationals, $m$ is a nonnegative integer, and $[q]$ denotes the list of rationals appearing in the product of transforms.

Given such a representation, each iteration can be efficiently performed in closed form. 
Specifically, each time we take integrals of the form
\begin{align*}
u_{q}(w) = e^{-\tfrac{k w \lambda}{q}} \int_0^w \frac{e^{\tfrac{k \lambda y}{q}} \big(k \lambda \, \E[z_{c,r,m,[q]}(S+y)]\big)}{q}\,dy,
\end{align*}
as in \Cref{eqn:isqk_constant,eqn:isqk_affine},
the result produces two terms of exactly the same structure. 
Thus, it is feasible to build a computational tool to iteratively generate these terms and extend the analysis to $k \approx 20$, 
since the branching factor of the recursion is at most 2. 
We leave the implementation of such a computational program to future work.

\section{Empirical and Numerical Results}
\label{sec:empirical}

In this section, we demonstrate the effectiveness of our method by presenting empirical and numerical results for our novel lower bounds on the mean response time in the $M/G/k$ system under arbitrary scheduling policies. We focus on our performance metric, the Uncertainty Improvement Ratio (UIR) which we define in \Cref{sec:empirical_uir}, and examine its behavior under a variety of settings.

We investigate the impact of changing the number of servers in \Cref{sec:varying_k}, of job size variability in \Cref{sec:varying_jobsizes}, and finally compare the UIR of our lower bounds against recently proposed upper bounds in the literature in \Cref{sec:upper_bound}.

\subsection{Uncertainty Improvement Ratio}
\label{sec:empirical_uir}

To quantify the quality of our lower bounds on mean response time (\Cref{thm:main_2,thm:main_k}) we use the Uncertainty Improvement Ratio (UIR).

This metric captures the shrinkage of the uncertainty region where the optimal policy may lie.  It measures how much of this region is closed by a new lower bound relative to the best prior lower bound. 
\begin{definition}
\label{def:uir}
The Uncertainty Improvement Ration of a response time lower bound $\E[T_{\text{novel-lower}}]$ relative to a pair of prior response time lower and upper bounds is given by 
    \begin{align*}
\mathrm{UIR} = \frac{\E[T_{\text{novel-lower}}] - \E[T_{\text{lower}}]}{\E[T_{\text{upper}}] - \E[T_{\text{lower}}]}.
\end{align*}
\end{definition}

Here $\E[T_{\text{lower}}]$ is a prior lower bound and $\E[T_{\text{upper}}]$ is a prior upper bound on the response time of $M/G/k$. We will use the simulated SRPT-$k$ response time as the upper bound. By definition, $\mathrm{UIR} \leq 100\%$, with $\mathrm{UIR}=0\%$ indicating no improvement over the prior bound 
and $\mathrm{UIR}=100\%$ indicating full closure of the uncertainty region.

\subsection{Varying number of servers $k$}
\label{sec:varying_k}

In this section, we examine the UIR of the ISQ-$k$ system. We study how the UIR of the MixEx bound relative to the naive lower bounds changes as $k$ increases. We also examine how the UIR of our novel bounds relative to the MixEx bound varies with $k$. Computing the response time requires numerically solving \Cref{eqn:main_k_1,eqn:main_k_2}. The explicit expressions for the functions $v_1$ and $u_1$ appearing in those equations are given by \Cref{eqn:isq2_constant,eqn:isq2_affine} for $k=2$ and also in \Cref{apdx:isq3} for $k = 3, 4,$ and $5$.
This numerical integration, performed using \texttt{Mathematica} or \texttt{SciPy}, does not present computational difficulties. We discuss how these functions could be computed for larger $k>5$ using a specialized computational program in \Cref{sec:generalization}.

\begin{figure}[htbp!]
    \centering
    \includegraphics[width=0.85\textwidth]{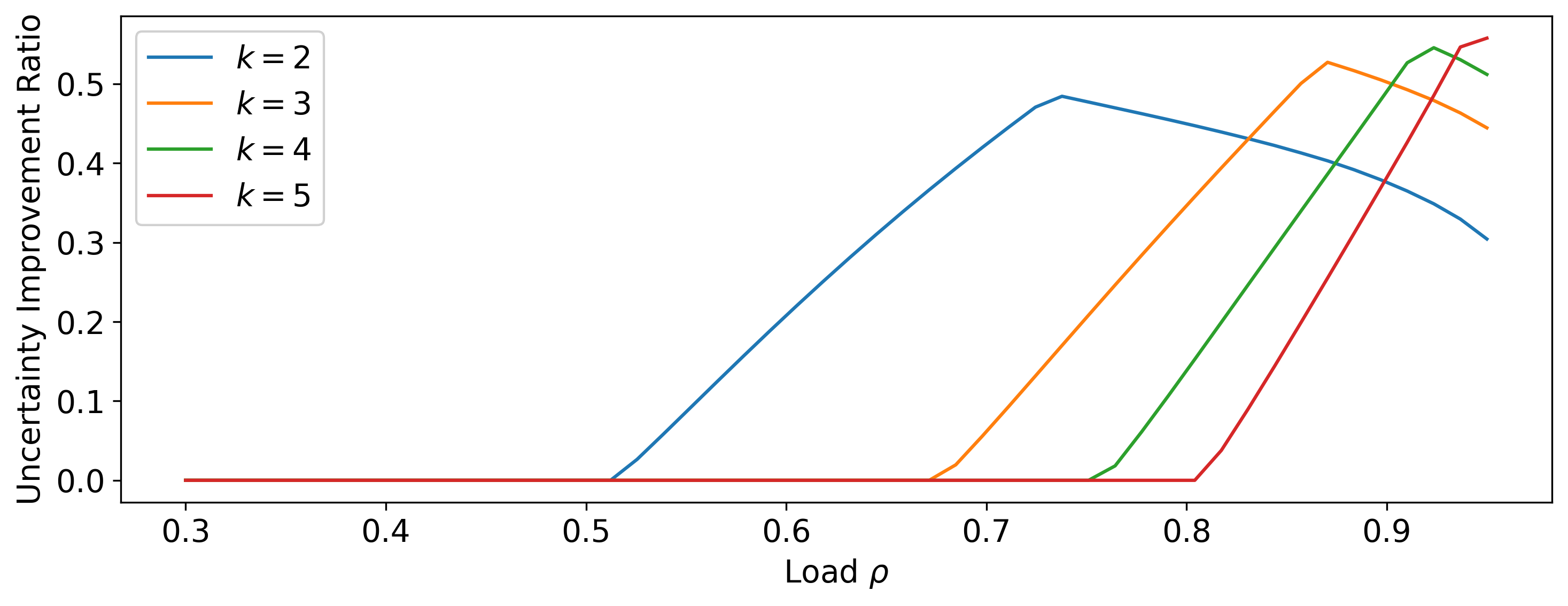}
    \caption{UIR of MixEx bound relative to naive bounds in the $M/G/k$ setting with exponential job sizes. Load ranges from 0.3 to 0.95 with five million simulated arrivals per load.}
    \label{fig:mixexk_uir}
\end{figure}

The UIR of the MixEx bound relative to the naive lower bounds is shown in \Cref{fig:mixexk_uir}. 
For $k = 2, 3, 4$ and $5$, we observe that the maximum UIR of the MixEx bound increases as $k$ increases. However, the improvement only begins to deviate from zero at progressively higher loads levels as $k$ increases, with the transition occurring approximately at $\rho \approx 1 - 1/k$. 

\begin{figure}[htbp!]
    \centering
    \includegraphics[width=0.85\textwidth]{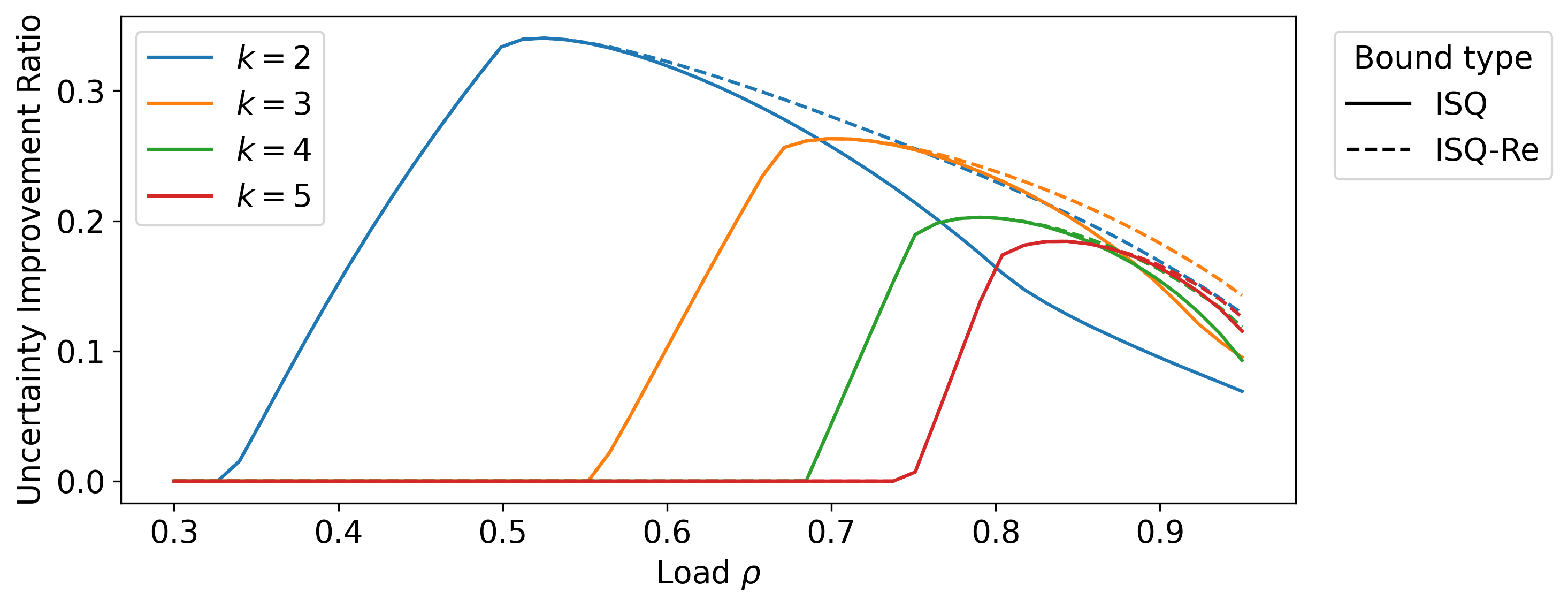}
    \caption{UIR of novel bounds relative to MixEx bound in the $M/G/k$ setting with exponential job sizes. Load ranges from 0.3 to 0.95 with five million simulated arrivals per load.}
    \label{fig:isqk_uir}
\end{figure}

The UIR of our novel lower bounds relative to the MixEx bound is shown in \Cref{fig:isqk_uir}. 

These bounds begin closing the uncertainty gap at significantly lower loads than MixEx and sustain meaningful improvement even as $\rho$ approaches 0.95. However, in contrast to the MixEx bound, the maximum UIR of the ISQ and ISQ-Recycling bounds decreases as $k$ increases. Note that, \Cref{fig:isqk_uir} should be viewed as an improvement upon the improvements already shown in \Cref{fig:mixexk_uir}. Taken together, the two bounds result in a maximum UIR improvement of roughly 60\% for all $k=2$ to 5, resulting in a substantial reduction in the uncertainty region at the optimal policy. This improvement is demonstrated across a wide range of load.

While the ISQ and ISQ-Recycling bounds are comparable for lower loads in this setting, the ISQ-Recycling bound provides a significant advantage at higher loads. As we show in the next subsection, this advantage becomes even more pronounced when the job size distribution has large variability (see \Cref{fig:hyper_uir}).

\subsection{Job size variability}
\label{sec:varying_jobsizes}

In this section, we study the improvements of the ISQ and ISQ-Recycling bounds over the MixEx bound for two families of job size distributions with differing variability, in the setting of two servers. 

In \Cref{fig:uniform_uir}, we consider a low variability case with three uniform job size distributions, where the squared coefficient of variation is $C^2 = \mathrm{Var}(S)/\E[S]^2 \in \{0.01, 0.05, 0.2\}$. In \Cref{fig:hyper_uir}, we examine a high variability case with three two-branch hyperexponential job size distributions, where $C^2 \in \{2, 3, 5\}$.

\begin{figure}[htbp!]
    \centering
    \includegraphics[width=0.85\textwidth]{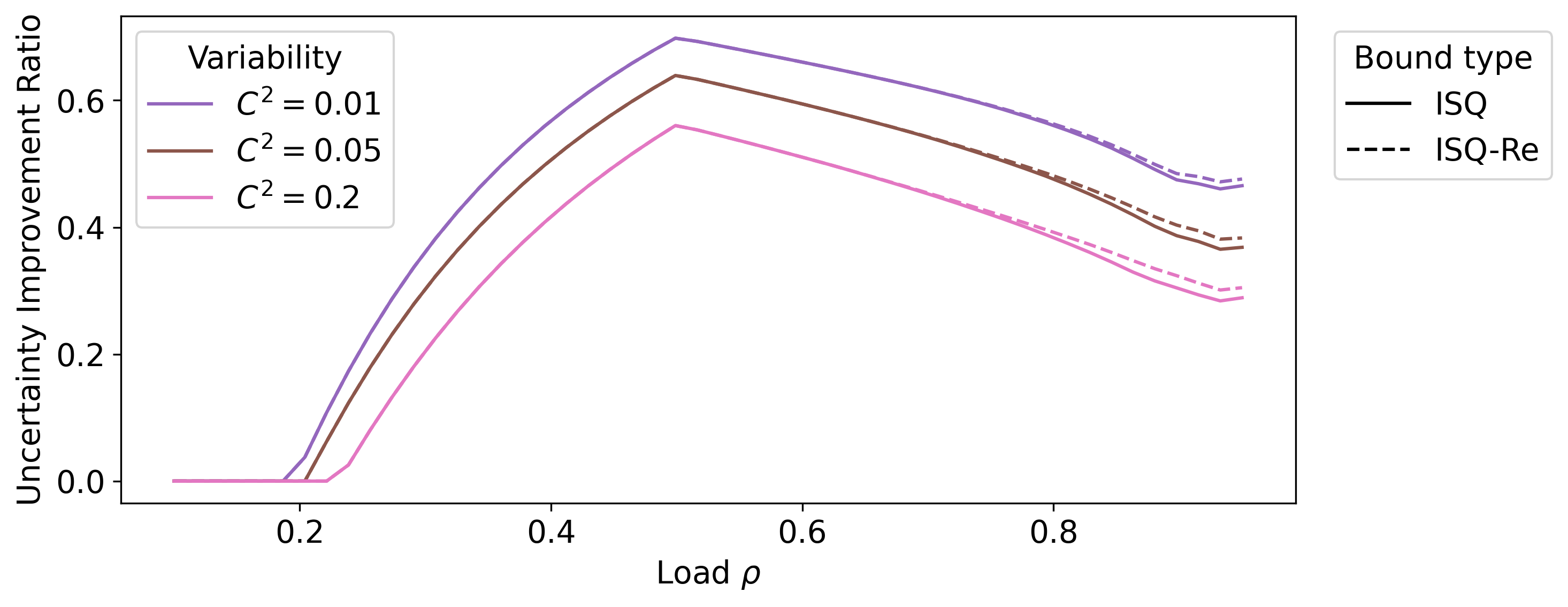}
    \caption{UIR relative to the MixEx bound for the ISQ and ISQ-Recycling bounds in the $M/G/2$ setting with Uniform job size distributions of various $C^2$. Load ranges from 0.1 to 0.95 with five million simulated arrivals per load.}
    \label{fig:uniform_uir}
\end{figure}

Important observations become apparent in both the high and low variability settings.
First, for the low variability distribution, the ISQ-$k$ based lower bounds achieve substantially greater improvement over MixEx, reaching a UIR of 62\% for the uniform distribution with $C^2 = 0.05$. Moreover, the improvement spreads broadly across the load range. Finally, the UIR continues to rise as the job size distribution approaches a deterministic distribution, approaching a limit just above 70\%.

An intuitive explanation is that, in low variability settings, busy periods are short and have similar lengths. Without large jobs to trigger long busy periods, the ISQ-$k$ system spends most of its time operating below a speed of 1. This makes ISQ-$k$ behave more similarly to a true $M/G/k$ system, allowing the ISQ and ISQ-Recycling bounds to yield substantially larger improvements over MixEx.

\begin{figure}[htbp!]
    \centering
    \includegraphics[width=0.85\textwidth]{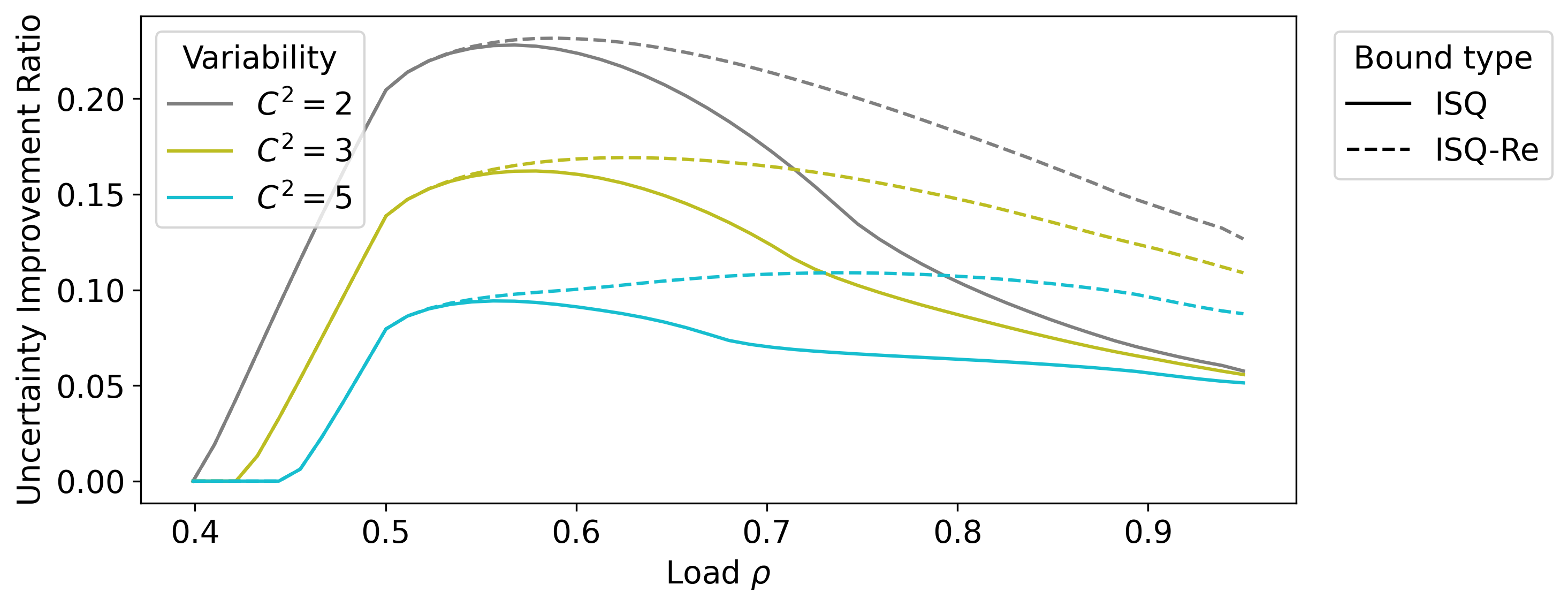}
    \caption{UIR relative to the MixEx bound for the ISQ and ISQ-Recycling bounds in the $M/G/2$ setting with hyperexponential job size distributions of various $C^2$. Load ranges from 0.4 to 0.95 with five million simulated arrivals per load.}
    \label{fig:hyper_uir}
\end{figure}

Secondly, for the high variability job size distributions, the Rec-ISQ-2 based ISQ-Recycling bound which incorporates recycling information offers a much greater improvement over MixEx than ISQ bound alone, justifying the need for \Cref{thm:ar_2}. This makes sense, as the larger, recycled jobs make up more of the relevant work at intermediate thresholds under job size distributions with high variability.
We also observe that the maximum UIR decreases as variability increases and the peak shifts to progressively higher load levels.

\subsection{Upper bound comparison}
\label{sec:upper_bound}
In this subsection, we compare the UIR of our novel lower bounds relative to the MixEx bound against the UIR of a recently introduced scheduling policy.

This upper bound corresponds to the the recently introduced SRPT-Except-$k+1$ (SEK) policy (\citet{Grosof2024BoMS} and \citet{Grosof2026}), which is proven to outperform SRPT-$k$ at all load levels and aims to tighten upper bounds on the mean response time of the optimal $M/G/k$ scheduling policy.

In particular, we empirically study the UIR of the Practical SEK policy, a variant of the core SEK policy which has been empirically shown to improve on SRPT-$k$.

We consider the setting of two servers with an exponential job size distribution. As shown in \Cref{fig:sek_uir}, uur best lower bound, the ISQ-Recycling bound, achieves a maximum UIR of approximately 34\%, whereas the Practical-SEK policy attains a UIR of only about 3\%. See \Cref{apdx:sek} for additional details on the Practical-SEK policy and the computational experiments.

Thus, our lower bounds close the uncertainty region far more dramatically from below than Practical-SEK does from above.

\begin{figure}[htbp!]
    \centering
    \includegraphics[width=0.85\textwidth]{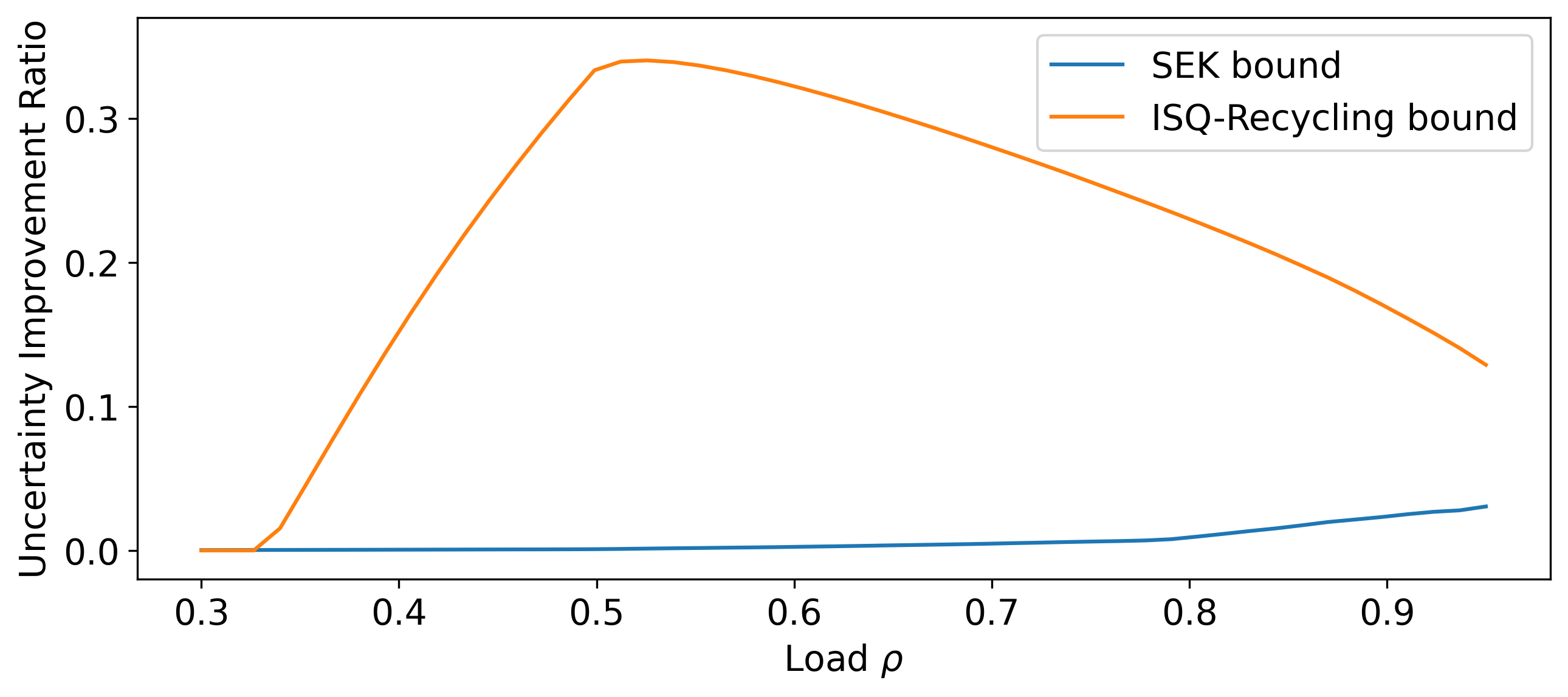}
    \caption{UIR relative to the SRPT-$k$ and the MixEx bound for the ISQ-Recycling and Practical-SEK bounds in the $M/G/2$ setting with $\text{Exp}(1)$ job sizes. Load ranges from 0.3 to 0.95 with five million simulated arrivals per load.}
    \label{fig:sek_uir}
\end{figure}

\section{Conclusion}
\label{sec:conclusion}

We present the first nontrivial lower bounds on the mean response time of the $M/G/k$ system under arbitrary scheduling policies. 

We start by applying the WINE combination to existing lower bounds, producing the MixEx bound, which substantially tightens the uncertainty region of the optimal policy, though it remains loose at intermediate loads.

Our main contributions are the ISQ-$k$ system and the DiffeDrift method, which together yield two novel lower bounds on response time, the ISQ and ISQ-Recycling bounds. Our ISQ and ISQ-Recycling bounds achieve significantly tighter results than MixEx bound.

Empirically, our lower bounds improve upon the prior naive bounds and the MixEx bound across a wide range of loads, with the most significant improvement appearing under moderate load conditions. For exponential job size distributions, the ISQ-Recycling bound achieves a maximum Uncertainty Improvement Ration (UIR) of 34\% on top of the MixEx bound, which itself provides a substantial improvement over the naive bounds. Moreover, for high variability job sizes, the ISQ-Recycling bound exhibits substantially higher UIR than the ISQ bound.

Future directions include proving \Cref{conj:x2} which concerns the behavior of jumps due to recycling for $k > 2$ servers, extending our framework to the setting of unknown job sizes as discussed in \Cref{sec:generalization}, generalizing to the $G/G/k$ model, and analyzing the asymptotic scaling of our bounds as $k$ grows large.

\bibliographystyle{plainnat}

\bibliography{reference.bib}

\appendix

\section{Proofs for Section 7}
\label{apdx:sec7}
\propdrift*

\begin{proof}
We have  $\int_\mathcal{S}g(x)\pi(dx)=\int_\mathcal{S}\int_\mathcal{S}g(y)P(t,x,dy)\,\pi(dx).$ This equality holds because $\pi$ is a stationary distribution, so $\int_\mathcal{S} P(t,x,U)\, \pi(dx) = \pi(U)$ for any measurable set $U$. Subtracting the LHS, dividing both sides by $t$, and introducing a limit,
\begin{align*}
    \nonumber
    0&=\dfrac{1}{t}\int_\mathcal{S}\left( \int_\mathcal{S} g(y)P(t,x,dy)-g(x)\right)\pi(dx)=\lim_{t\to0}\int_\mathcal{S} \frac{1}{t} \left(\int_\mathcal{S} g(y)P(t,x,dy)-g(x)\right)\pi(dx).
\end{align*}
Therefore,
\begin{align}
    \label{eq:general-bar}
    0&=\lim_{t\to0}\int_\mathcal{S}  A(t,x)\,\pi(dx)+\lim_{t\to0}\int_\mathcal{S}B(t,x)\,\pi(dx)
\end{align}
For the first term on the RHS, we invoke the dominated convergence theorem (DCT) in order to take the limit under the integral. The assumption of the DCT is satisfied by the assumption that $\lim_{t \to 0} A(t,\cdot)\to G\circ g(\cdot)$ uniformly, so there exists some $t$ small enough such that the integrand can be bounded by, e.g., $2|G\circ g(x)|$ for all $x\in \mathcal{S}$, which is $\pi$-integrable by assumption.

The second term also vanishes by assumption. Therefore, \eqref{eq:general-bar} is equivalent to the following:
\begin{align*}
    0=\int_\mathcal{S}\lim_{t\to0} A(t,x)\pi(dx)+0=\int_\mathcal{S} G\circ g(x)\pi(dx),
\end{align*}
giving us the desired result. 
\end{proof}

\lemwsquare*
\begin{proof}
    Imagine an $M/G/1$ queue such that after each busy period, the servers stops working until the system accumulates $k$ jobs, then restarts and starts processing jobs at a rate of 1. We call this system an $M/G/1$ queue with server activation threshold.
    Then the expected work of an ISQ-$k$ queue will be bounded above by this $M/G/1$ queue with a server activation threshold. The expected work in this $M/G/1$ queue with a server activation threshold is equivalent to the expected waiting time of an $M/G/1$-FCFS queue with rest periods drawn from a distribution $T_0$, where $T_0$ is the sum of $k$ exponential random variables, and $T_0$ is coupled to the arrival process.

    $M/G/1$ queues with rest periods are studied in \citet{Scholl1983O}. \citet[Equation (15)]{Scholl1983O} characterize the expected second moment of waiting time of an $M/G/1$ queue with rest periods. In particular, this expectation is finite if $E[S^3]<\infty$ and $\E[T_0^3]<\infty$. Note that the sum of $k$ exponential distributions has finite moments of all orders, so the lemma holds by our assumption that $\E[S^3]<\infty$. Note that \citet{Scholl1983O} assume that $T_0$ is independent of the arrival process, which is not the case in our setting. However the proof of \cite[Equation (15)]{Scholl1983O} proceeds by first characterizing the number of arrivals during the rest period, before proceeding with the rest of the proof. In our system, that number of arrivals always $k$, simplifying \citet[Equation (7)]{Scholl1983O}, and remainder of the proof holds unchanged.
 \end{proof}

\lemuniform*
\begin{proof}
To illustrate the high level idea, we first consider the case where the speed $i$ is 1. By the definition of the Poisson arrival process, for $t < w$,
\begin{align}
    \nonumber
    &\dfrac{1}{t}\E[g(W(t),1)-g(w,1)\,|\,W(0)=w, I(0)=1]\\
    \nonumber
    &=\mathbb{P}(N_t=0)g(w-t,1)+\mathbb{P}(N_t=1)\E[g(w-t+S,1)]+\mathbb{P}(N_t>1)\E[g(w-t+N_tS,1)\mid N_t>1]\\
    \nonumber
    &=\dfrac{1}{t}(1-\lambda t)g(w-t,1)+\frac{o(t)}{t}g(w-t,1) +\dfrac{1}{t}\lambda t \E[g(w-t+S,1)]+\frac{o(t)}{t}\E[g(w-t+S,1)]\\
    \nonumber
    &\quad + \frac{o(t)}{t}\E[g(w-t+N_t\cdot S,1)\mid N_t>1]\\
    \label{eq:uniform-poisson}
    &=\underbrace{-\frac{g(w,1)}{t}+\frac{g(w-t,1)}{t}+\lambda\E[g(S+w-t,1)]-\lambda g(w-t,1)}_{A(t,w,1)}+B(t,w,1)
\end{align}
Before we continue, note that to ensure uniform convergence, we must also consider the case that $t > w$, for $w$ close to 0. In this case, the only significant change is the at the $N_t = 0$ term results in a state after time $t$ of $(0, 0)$, rather than $(w-t, 1)$. By assumption, $\lim_{w \to 0^+} g(w, 1) = g(0, 0)$, so the proof proceeds unchanged.

Returning to \eqref{eq:uniform-poisson}, note that the $A(t,w,1)$ term contains the terms corresponding to the scenario where there is only one arrival during the period of length $t$ and the scenario where there is no arrival during the same period. $B(t,w,1)$ corresponds to the remaining negligible terms:
\begin{align*}
    B(t,w,1)=\frac{o(t)}{t}g(w-t,1)+\frac{o(t)}{t}\E[g(w-t+S,1)]+\frac{o(t)}{t}\E[g(w-t+N_t\cdot S,1)\mid N_t>1].
\end{align*}
By the independence of $S$ and $N_t$, it is easy to check that $\E[g(w-t+N_t\cdot S,1)\,|\, N_t>1]<\infty$ given the assumption that $g(w,i)=w^2+c(w,i)$ where $|c(w,i)|\leq C_1w+C_2$. Therefore, we have, $\lim_{t\to0} \E[B(t,W,1)]=0.$ We now want to show that $A(t,\cdot,1)$ converges uniformly to some limit as $t \to 0$, namely, $G\circ g$. We first consider the term $-g(w,1)/t+g(w-t,1)/t$. Since $g$ is twice-differentiable for all $w\in\mathbb{R}_{+}$, by Taylor expansion,
\begin{align}
\label{eqn:uniform_1}
    g(w,1)-g(w-t,1)=g'(w,1)t-R_1(w,t)=g'(w,1)t-g''(c,1)\frac{t^2}{2},
\end{align}
where $R_1$ is some suitable remainder term. Recall that by our assumption, $c(w,i)$ has a bounded second derivative and $g(w,i)$ also has a bounded second derivative. Therefore, $\dfrac{1}{t}(-g(w,1)+g(w-t,1))=-g'(w,1)+\frac{g''(c,1)}{2}t\to -g'(w,1)$ uniformly for all $w$. Next, we want to show that $\lambda \E[g(S+w-t)]-\lambda g(w-t)\to \lambda\E[g(S+w)]-\lambda g(w)$ uniformly. Let $\hat{g}(w-t,1)$ be defined as this difference, namely $\E[g(S+w-t,1)]-g(w-t,1)$. Then again since $\hat{g}$ is smooth, we have by Taylor's theorem
\begin{align*}
\label{eqn:uniform_2}
    |\hat{g}(w-t,1)-\hat{g}(w,1)|=|-\hat{g}'(w,1)t+R_1(w,t)|\leq|\hat{g}'(w,1)|t+M_2\frac{t^2}{2},
\end{align*}
where $R_1(w, t)$ is the remainder term of the Taylor expansion, for the second-order and higher terms. Therefore the convergence is uniform if $|\hat{g}'(w,1)|$ is bounded. We have
\begin{align}
    \E[g(S+w-t,1)-g(w-t,1)]&=\E[S^2]-2\E[S]t+2\E[S]w+\hat{c}(w,t),
\end{align}
where $\hat{c}(w,t)$ is defined analogously to $\hat{g}$.
By our assumption that $|c'(w, i)| \le M_1$, the RHS has a bounded derivative, thus establishing the uniform convergence.

When the speed of the system is less than 1, the only necessary changes lie in \Cref{eqn:uniform_1,eqn:uniform_2}. \Cref{eqn:uniform_1} becomes $g(w,i/k)-g(w-it/k,i/k)=g'(w,i/k)it/k-\bar{R}_1(w,t),$ for some remainder term $\bar{R}_1$. We now invoke our bounded second derivative assumption to bound $\bar{R}_1(w,t)$ in the same manner as $R_1(w, t)$ above.
Dividing by $t$ and taking the limit as $t \to 0$ yields the desired limit $g'(w,i/k)\frac{i}{k}$. Similarly, \Cref{eqn:uniform_2} can now be written as, 
\begin{align*}
    \E[g(S+w-t,(i+1/k))-g(w-t,i/k)]&=\E[S^2]-2\E[S]\bar{i}t/k+2\E[S]w+\bar{c}(w,t),
\end{align*}
using Taylor's theorem and the intermediate value theorem,
for some value $\bar{i}\in[i,i+1]$ and some function $\bar{c}(w,t)$ defined analogously to $\hat{c}$ with a bounded first derivative by our assumption. Therefore as $t\to 0$, $\E[g(S+w-t,(i+1/k))-g(w-t,i/k)]$ converges uniformly to some limit, establishing uniform convergence. 

\end{proof}

\section{Proofs for Section 9}
\label{apdx:sec9}
In this appendix, we provide proofs of lemmas needed in \Cref{sec:mgk} and \Cref{thm:ar_k}.
\begin{lemma}
\label{lem:isqk_constant_assumptions}
    $g_k(\cdot, \cdot)$ defined by \Cref{eqn:isqk_constant} satisfies the assumptions of \Cref{prop:drift} and \Cref{lem:uniform}. In particular, $\lim_{w\to0^+}g_k(w,i)=0$. Note also that $g_k(\cdot,\cdot) \geq 0$.
\end{lemma}
\begin{proof}
By \Cref{eqn:isqk_constant} we see that $u_q(0)=0$. Moreover, as we are integrating a non-negative function recursively starting from $u_k(\cdot) = 0$, each of $u_q(\cdot)$ must be non-negative, i.e., $u_q(\cdot)\geq 0$ for all $q\in\{1,...,k\}$.

Moreover, one can see that from the form of \Cref{eqn:isqk_constant} that each $u_q$ is a sum of linear and negative exponential terms in $w$. In particular, in the integrand, $\E[u_{q+1}(S+y)]$ preserves this structure during each recursive step, see, e.g., the ISQ-3 test functions given in \Cref{apdx:isq3}. Therefore, $u_q$ must satisfy the assumptions of \Cref{prop:drift} and \Cref{lem:uniform}.

\end{proof}

Applying \Cref{lem:drift} to \Cref{eqn:isqk_constant} we have,

\begin{lemma}
\label{lem:isqk_constant}
     For any arbitrary job size $S$ such that $E[S^3]$ is finite, the drift at speed 0 is $G\circ g_k(0,0)=\lambda \E[S+u_1(S)]$ and the drift at all other speeds $i \geq 1/k$ is given by $G\circ g_k(w,i)=\lambda\E[S]-1$ for all $i\geq 1/k$.
\end{lemma}
\begin{proof}
    Clearly, the drift at speed $1$ is given by $G\circ g_k(w,1)=\lambda\E[S]-1$ and the drift at speed $0$ is given by $G\circ g_k(0,0)=\lambda\E[S+u_1(S)]$. For all other speed $q/k$, the drift is given by,
    \begin{align*}
        G\circ g_k(w,q/k)&=\lambda\left((w+\E[S]+\E[u_{q+1}(w+S)])-(w+u_q(w)) \right)-(w+u_q(w))'\cdot \dfrac{q}{k}\\
        &=\lambda\E[S]-1+1-\frac{q}{k}-\lambda u_q(w)+\lambda\E[u_{q+1}(w+S)]-\dfrac{qu_q'(w)}{k}.
    \end{align*}
    We want to show that the above is equal to $\lambda\E[S]-1$. Therefore, we want prove that the set of $u_q(w)$ functions collectively solve the following first order linear ordinary differential equations:
    \begin{align*}
        1-\frac{q}{k}-\lambda u_q(w)+\lambda\E[u_{q+1}(w+S)]-\dfrac{q\cdot u_q'(w)}{k}=0,
    \end{align*}
    with initial condition $u_q(0)=0$.
    Solving this differential equation with $u_k(w)=0$ yields the formula for $u_{k-1}(w)$ in \Cref{def:isqk_contant}, and solving for decreasing $q$ yields all of the functions $u_q(w)$.

\end{proof}

Note that the differential equations used to define the constant-drift and affine-drift test functions $g_k$ and $h_k$ are generalizations of the differential equations used to define the constant and affine-drift test functions $g_2$ and $h_2$ for the 2-server case in \Cref{sec:derive_affine}.

Via the same argument \Cref{lem:isqk_constant_assumptions}, we prove the following lemma, which allows us to apply \Cref{lem:drift}.

\begin{lemma} 
\label{lem:isqk_affined_bound}
    $h_k(\cdot,\cdot)$ defined by \Cref{eqn:isqk_affine} satisfies the assumptions of \Cref{prop:drift} and \Cref{lem:uniform}. In particular, $\lim_{w \to 0^+} h_k(w,i)=0$. Note also that $h_k(\cdot,\cdot)\geq0$.
\end{lemma}

Now, we examine the drift of the test function $h_k$.

\begin{lemma}
\label{lem:isqk_affine}
    For any arbitrary job size $S$ such that $E[S^3]$ is finite, the drift at speed $0$ is given by $G \circ h_k(0, 0) = \lambda\E[S^2+v_1(S)]$ and the drift at all other speed $\geq 1/k$ is given by $G\circ h_k(w,i)=\lambda\E[S^2]+2w(\lambda\E[S]-1)$.
\end{lemma}
\begin{proof}
     Clearly, the drift at speed $1$ is given by $G\circ h_k(w,1)=\lambda\E[S^2]+2w(\lambda\E[S]-1)$ and the drift at speed $0$ is given by $G\circ h_k(0,0)=\lambda\E[S^2+v_1(S)]$. For all other speed $q/k$, the drift is given by,
     \begin{align*}
         G\circ h_k(w,q/k)&=\lambda\left(((w+\E[S])^2+\E[v_{q+1}(w+S)] )-(w^2+v_{q}(w))\right)-\left(w^2+v_q(w) \right)'\cdot\frac{q}{k}\\
         &=\lambda\E[S^2]+2w(\lambda\E[S]-1)+2w-\frac{2qw}{k}-\lambda v_q(w)+\lambda \E[v_{q+1}(S+w)]-\frac{qv'_q(w)}{k}.
     \end{align*}
     We want to show that the above is equal to $\lambda\E[S^2]+2w(\lambda\E[S]-1)$. Therefore, we want prove that the set of $v_q(w)$ functions collectively solve the first order linear ordinary differential equations $2w-\frac{2qw}{k}-\lambda v_q(w)+\lambda \E[v_{q+1}(S+w)]-\frac{qv'_q(w)}{k}=0,$ with initial condition $v_q(0)=0$.
    Solving this differential equation with $v_k(w)=0$ yields the formula for $v_{k-1}(w)$ in \Cref{def:isqk_affine}, and solving for decreasing $q$ yields all of the functions $v_q(w)$. 

\end{proof}

The next set of lemmas concerns the modified affine-drift test function defined in \Cref{def:isqk_modified}. By the same reasoning as \Cref{lem:isqk_constant_assumptions} we have the following lemma. 

\begin{lemma}
\label{lem:isqk_modified_bound}
    $h_{k,x}(\cdot, \cdot)$ defined by \Cref{eqn:isqk_modified} satisfies the assumptions of \Cref{prop:drift} and \Cref{lem:uniform}.
\end{lemma}

The above lemma allows us to apply \Cref{lem:drift} to \Cref{eqn:isqk_modified}.
\begin{lemma}
\label{lem:isqk_modified}
    For any arbitrary truncated job size $S_x$ and arrival rate $\lambda_x$, the drift at $h_{k,x}$ is given by
    \begin{align*}
        G \circ h_{k,x}(w,1) &=\lambda_x\E[S_x^2]+2w(\lambda_x\E[S_x]-1) \\    
        \forall 0 < i < 1,
        G\circ h_{k,x}(w,i) &=\lambda_x\E[S_x^2]+2w(\lambda_x\E[S_x]-1)+\frac{k-i}{k}C_k(x,\lambda_x) \\
        G\circ h_{k,x}(0,0)&=\lambda_x(\E[S_x^2]+\E[\ell_1(S_x)]).
    \end{align*}
\end{lemma}

\begin{proof}
    Clearly, the drift at speed 1 is given by $G\circ h_{k,x}(w,1)=\lambda_x\E[S_x^2]+2w(\lambda_x\E[S_x]-1)$ and the drift at speed 0 is given by $G\circ h_{k,x}(0,0)=\lambda_x\E[S_x^2+\ell_1(S_x)]$. For all other speeds $i=q/k$, the drift is,
    \begin{align*}
        G\circ h_{k,x}(w,q/k)&=\lambda_x\left(((w+\E[S_x])^2+\E[\ell_{q+1}(w+S_x)] )-(w^2+\ell_{q}(w))\right)-\left(w^2+\ell_q(w) \right)'\cdot\frac{q}{k}\\
         &=\lambda_x\E[S_x^2]+2w(\lambda_x\E[S_x]-1)+2w-\frac{2qw}{k}-\lambda_x \ell_q(w)+\lambda_x \E[\ell_{q+1}(S_x+w)]-\frac{q\ell'_q(w)}{k}.
    \end{align*}
    We want to show that the above is equal to $\lambda_x\E[S_x^2]+2w(\lambda_x\E[S_x]-1)+(k-q)C_k(x,\lambda_x)$.  Therefore, we want to prove that the set of $\ell_q(w)$ functions collectively solve the following differential equations:
    \begin{align*}
        2w-\frac{2qw}{k}-\lambda_x \ell_q(w)+\lambda_x \E[\ell_{q+1}(S_x+w)]-\frac{q\ell'_q(w)}{k}=(k-q)C_k(x,\lambda_x),
    \end{align*}
    with initial condition $\ell_q(0)=0$. Solving this differential equation with $\ell_k(w)=0$ yields the formula for $\ell_{k-1}(w)$ in \Cref{def:isqk_modified}, and solving for decreasing $q$ yields all of the functions $l_q(w)$.

\end{proof}

In the above lemma, the drift at speed 0 is given by $\lambda_x(\E[S_x^2]+\E[\ell_1(S_x)])$. However, because the recursive nature of $\ell_q$ makes $\ell_1$ very hard to work with, we now prove a lemma which gives us an alternative expression for $\E[\ell_1(S_x)]$.
\begin{lemma}
\label{lem:isqk_C}
   $C_k(x,\lambda_x)$ defined by \Cref{eqn:C_k} solves the following equation $\lambda_x\E[\ell_1(S_x)]=kC_k(x,\lambda_x)$.
\end{lemma}

\begin{proof}
    By \Cref{lem:isqk_modified}, the drift of $h_{k,x}$ at speed $i>0$ is given by $G\circ h_{k,x}(w,i)=\lambda_x\E[S_x^2]+2w\left(\rho_x-1 \right)+C_k(x,\lambda_x)\frac{k-i}{k}$ and the drift at speed 0 is given by $G\circ h_{k,x}(0,0)=\lambda_x\E[S_x^2]+\lambda_x\E[\ell_1(S_x)]$. Therefore, taking expectation over $W_x$ and $I$, we have,
    \begin{align*}
        &\E[G\circ h_{k,x}(W_x,I)] =\lambda\E[S_x^2]+2\E[W_x]\left(\rho_x-1 \right)+k C_k(x)\sum_{i=1}^k\dfrac{k-i}{k}\mathbb{P}(I=i/k)+\lambda\E[\ell_1(S_x)]\mathbb{P}(I=0/k)\\
        &= \lambda\E[S_x^2]+2\E[W_x]\left(\rho_x-1 \right)+k C_k(x)\sum_{i=0}^k\dfrac{k-i}{k}\mathbb{P}(I=i/k)+(\lambda\E[\ell_1(S_x)]-kC_k(x,\lambda))\mathbb{P}(I=0/k).
    \end{align*}
    To simplify the above above equation, we apply \Cref{lem:drift} to the test function $g(w,i)=w$ for all $(w,i)$. We have, $1-\rho_{x}=\sum_{i=0}^k\frac{k-i}{k}\mathbb{P}(I=i/k)$. Therefore, 
    \begin{align*}
        \E[G\circ h_{k,x}(W,I)]=\lambda\E[S_x^2]+2\E[W_x]\left(\rho_x-1 \right)+k C_k(x)(1-\rho_{x})+ (\lambda\E[\ell_1(S_x)]-kC_k(x,\lambda))\mathbb{P}(I=0/k)
    \end{align*}
    
    By \Cref{lem:isqk_modified_bound}, the ISQ-$k$ modified affine-drift test function defined in \Cref{def:isqk_modified} satisfies the assumption of \Cref{lem:drift}. We have $\E[G\circ h_{k}(W_x,I)]=0$. Solving for $\E[W_x]$ we get,
    \begin{align}
    \label{eqn:isqk_work_modified}
        \E[W_x]=\dfrac{\lambda_x\E[S_x^2]}{2(1-\lambda_x\E[S_x])}+\dfrac{k}{2}\cdot C_k(x,\lambda_x)+\dfrac{(\lambda_x\E[\ell_1(S_x)]-kC_k(x,\lambda_x))\mathbb{P}(I=0/k)}{2(1-\rho_x)}.
    \end{align}
    Recall that $C_k(x,\lambda_x)$ is defined to be $C_k(x,\lambda_x)=\frac{2}{k}\cdot\frac{\lambda_x\E[v_1(S_x)]}{2+2\lambda_x\E[u_1(S_x)]}$.
    Recall also that by \Cref{prop:isqk_work}, we have the following alternative expression for mean relevant work for the same arrival rate $\lambda_x$ and truncated job size $S_x$:
    \begin{align}
         \E[W_x]=\dfrac{\lambda_x\E[S_x^2]}{2(1-\lambda_x\E[S_x])}+\dfrac{\lambda_x\E[v_1(S_x)]}{2+2\lambda_x\E[u_1(S_x)]}
         = \dfrac{\lambda_x\E[S_x^2]}{2(1-\lambda_x\E[S_x])}+\dfrac{k}{2}\cdot C_k(x,\lambda_x)
    \end{align}
    Thus, taking the difference of these two expressions for $\E[W_x]$, we find that
    the third term of \eqref{eqn:isqk_work_modified} is zero:
    \begin{align*}
        0 = \dfrac{(\lambda_x\E[\ell_1(S_x)]-kC_k(x,\lambda_x))\mathbb{P}(I=0/k)}{2(1-\rho_x)}
    \end{align*}
    
    Note that $\mathbb{P}(I=0/k)$ is positive, by \Cref{eqn:isqk_prob0}. Thus solving $0 = \lambda\E[\ell_1(S_x)]-kC_k(x,\lambda_x)$ we get
    $\lambda_x\E[\ell_1(S_x)] =kC_k(x,\lambda_x)$
    as desired.

\end{proof}

Therefore, we can write the drift at speed 0 as $G\circ h_{k,x}(0,0)=\lambda_x\E[S_x^2]+kC_k(x,\lambda_x)$. We are now ready to state the general version of the Rec-ISQ-$k$ lower bound. This is the second lower bound of \Cref{thm:main_k}.

\mainarisqk*

\begin{proof}
    The proof of this theorem follows the same argument as the proof of \Cref{thm:ar_2}. Let $\E_\mathfrak{r}[\cdot]$ denote the Palm expectation taken over the moments when the arbitrary arrival of jobs of size $x$ occurs according to the recycling stream.

    Let $J_x(w, i) := h_{k,x}(w+x,\min(i+1/k, 1)) - h_{k,x}(w,i)$
    denote the increase in the test function due to the arrival of a size-$x$ job.
    Then $\E_\mathfrak{r}[J_x(W_x,I)]$ denotes the mean size of the recycling jump in stationarity. Thus, by \Cref{lem:ar_drift}, $\E[G \circ h_{k,x}(W_x, I)] = \E[\text{Poisson drift}] + (\lambda - \lambda_x) \E_\mathfrak{r}[J_x(W_x,I)]$.
    By \Cref{lem:isqk_C} we can write the expected Poisson drift as,
    \begin{align*}
    \E[\text{Poisson drift}]&=\lambda_x\E[S_x^2]+2\E[W_x](-1+\rho_x)+kC_k(x,\lambda)\sum_{i=0}^k\frac{k-i}{k}\mathbb{P}(I^r=i/k)\\
    &=\lambda_x\E[S_x^2]+2\E[W_x](-1+\rho_x)+kC_k(x,\lambda)(1-\rho_{\bar{x}}).
    \end{align*}
    In the second equality, we used the fact that $1-\rho_{\overbar{x}}=\sum_{i=0}^k\frac{k-i}{k}\mathbb{P}(I^r=i/k)$, a result that can be derived similarly to \Cref{eqn:isq2_load} by applying \Cref{lem:ar_drift} to the test function $g(w,i)=w$ for all $(w,i)$.
    
    Therefore, substituting the above into $\E[G\circ h_{k,x}(W_x,I)]=0$ and solving for $\E[W_x]$, we have
    \begin{align*}
        \E[W_{x}^{\text{arb-}ISQ}]&=\frac{\lambda_x E[S_x^2]}{2 (1 - \rho_x)} + \dfrac{kC_k(x,\lambda_x)(1-\rho_{\overbar{x}})}{2(1-\rho_x)}+\frac{(\lambda - \lambda_x)\E_\mathfrak{r}[J_x(W_x,I)]}{2(1-\rho_x)}\\
        &=\dfrac{\lambda_x\E[S_x^2]}{2(1-\lambda_x\E[S_x])}+\dfrac{\lambda_x\E[v_1(S_x)]}{2+2\lambda_x\E[u_1(S_x)]}\cdot\dfrac{1-\rho_{\bar{x}}}{1-\rho_x}+\frac{(\lambda - \lambda_x)\E_\mathfrak{r}[J_x(W_x,I)]}{2(1-\rho_x)},
    \end{align*}
    where we use the definition of $C_k(x,\lambda)$ in the second equality. Recall the definition of an increasing speed queue. The current speed $i$ bounds the number of jobs in the system, which in turn bounds the total work, which consists of at most $ki$ jobs, each with size at most $x$. Therefore, for $0<i<1$, the system can only visit states $(w, i)$ where $w \le kix$.
    \begin{align*}
        J_x(w,i)=h_{k,x}(w+x,i+1/k)-h_{k,x}(w,i/k)\geq \inf_{w\in[0,k i x]}h_{k,x}(w+x,i+1/k)-h_{k,x}(w,i/k),
    \end{align*}
    along with $J_x(0,0)=h_{k,x}(x,1/k)$ and $J_x(w+x,1)\geq x^2$, we have $\E_\mathfrak{r}[J_x(W_x,I)]\geq J_x$.

\end{proof}

\section{List of ISQ-$k$ Test Functions}
\label{apdx:isq3}

In this appendix, we list the constant- and affine-drift test functions for the
ISQ-$k$ system with $k = 3, 4,$ and $5$.
The quantities $u_1$ and $v_1$ in \Cref{eqn:main_k_1,eqn:main_k_2} are obtained
directly from these test functions.
In particular,
\[
    u_1(w) = g_k(w,1/k) - w,
    \qquad
    v_1(w) = h_k(w,1/k) - w^2.
\]

\subsection{ISQ-3 Test Functions}

For the ISQ-3 system, the constant-drift test function $g_3$ defined in \Cref{def:isqk_contant} is
\begin{flalign*}
& g_3(0,0)=0, &&\notag\\
& g_3(w,1)=w, &&\notag\\
& g_3\!\left(w,\tfrac{2}{3}\right)
   = w+\frac{1}{3\lambda}\!\left(1-e^{-\frac{3}{2}w\lambda}\right), &&\notag\\
& g_3\!\left(w,\tfrac{1}{3}\right)
   = w+\frac{1}{\lambda}
     +\frac{\bigl(2\widetilde{S}(\tfrac{3}{2}\lambda)-3\bigr)e^{-3w\lambda}}{3\lambda}
     -\frac{2\widetilde{S}(\tfrac{3}{2}\lambda)e^{-\frac{3}{2}w\lambda}}{3\lambda}. &&
\end{flalign*}
The affine-drift test function $h_3$ defined in \Cref{def:isqk_affine} is
\begin{flalign*}
& h_3(0,0)=0, &&\notag\\
& h_3(w,1)=w^2, &&\notag\\
& h_3\!\left(w,\tfrac{2}{3}\right)
   = w^2+\frac{1}{9\lambda^2}\!\left(6w\lambda+4e^{-\frac{3}{2}w\lambda}-4\right) ,&&\notag\\
& h_3\!\left(w,\tfrac{1}{3}\right)
   = w^2+\frac{1}{9\lambda^2}\!\left(-10+10e^{-3w\lambda}-8\widetilde{S}\!\left(\tfrac{3}{2}\lambda\right)e^{-3w\lambda}+8\widetilde{S}\!\left(\tfrac{3}{2}\lambda\right)e^{-\frac{3}{2}w\lambda}\right)+\frac{1}{3\lambda}\!\left(2\E[S]-2e^{-3w\lambda}\E[S]+6w\right). &&
\end{flalign*}

\subsection{ISQ-4 Test Functions} For the ISQ-4 system, the constant-drift test function $g_4$ defined in \Cref{eqn:isqk_constant} is
\begin{flalign*}
& g_4(0,0)=0, &&\\
& g_4(w,1)=w, &&\\
& g_4(w,\tfrac{3}{4})
   = w+\frac{1}{4\lambda}\left(1-e^{-\frac{4w\lambda}{3}} \right), &&\\
& g_4(w, \tfrac{2}{4}) = w+\dfrac{3}{4\lambda}\left(1-e^{-2w\lambda}+\widetilde{S}(\tfrac{4}{3}\lambda)e^{-2w\lambda}-\widetilde{S}(\tfrac{4}{3}\lambda)e^{-\frac{4}{3}w\lambda}
\right),
\\
& g_4(w,\tfrac{1}{4}) = w + \dfrac{3}{2\lambda}\left(1-e^{-4w\lambda}+\widetilde{S}(2\lambda)e^{-4w\lambda}-\widetilde{S}(2\lambda)\widetilde{S}(\tfrac{4}{3}\lambda)e^{-4w\lambda}-\widetilde{S}(2\lambda)e^{-2w\lambda}+\widetilde{S}(\tfrac{4}{3}\lambda)\widetilde{S}(2\label{})e^{-2w\lambda}\right)  \\
&\qquad+\dfrac{9}{8\lambda}\left(\widetilde{S}(\tfrac{4}{3}\lambda)^2e^{-4w\lambda}-\widetilde{S}(\tfrac{4}{3}\lambda)^2e^{-\frac{4w\lambda}{3}} \right).
\end{flalign*}
The affine-drift test function $h_4$ defined in \Cref{def:isqk_affine} is 
\begin{flalign*}
    & h_4(0,0)=0, &&\\
    & h_4(w,1)=w^2,\\
    & h_4(w,\tfrac{3}{4})=w^2-\frac{3}{8\lambda^2}\left(e^{-\frac{4\lambda}{3}}\right)+\frac{w}{2\lambda},\\
    & h_4(w,\tfrac{2}{4})=w^2+\frac{9}{8\lambda^2}\left( -1+e^{-2w\lambda}-\widetilde{S}(\tfrac{4}{3}\lambda)e^{-2w\lambda}+\widetilde{S}(\tfrac{4}{3}\lambda)e^{-\frac{4w\lambda}{3}}\right)+\dfrac{3w+\E[S]-e^{-2w\lambda}E[S]}{2\lambda},\\
    &h(w,\tfrac{1}{4})=w^2+\frac{15}{8\lambda^2}\left(-1+e^{-4w\lambda} \right) +\frac{9}{4\lambda^2}\left(\widetilde{S}(2\lambda)e^{-4w\lambda}+9\widetilde{S}(\tfrac{4}{3}\lambda)\widetilde{S}(2\lambda)e^{-4w\lambda}+\widetilde{S}(2\lambda)e^{-2w\lambda}\right)\\
    &\qquad +\dfrac{27}{16\lambda^2}\left(\widetilde{S}(\tfrac{4}{3}\lambda)^2e^{-4w\lambda_x}+\widetilde{S}(\tfrac{4}{3}\lambda)^2e^{-\frac{4w\lambda}{3}} \right)+\dfrac{3w+2\E[S]-2e^{-4w\lambda}E[S]+\widetilde{S}(2\lambda)\E[S](e^{-4w\lambda}-e^{-2w\lambda})}{\lambda}
\end{flalign*}

\subsection{ISQ-5 Test Functions} For the ISQ-5 system, the constant-drift test function $g_5$ defined in \Cref{eqn:isqk_constant} is
\begin{flalign*}
    & g_5(0,0)=0, &&\\
    & g_5(w,1)=1,\\
    & g_5(w,\tfrac{4}{5})=w+\frac{1}{5\lambda}\left(1-e^{-\frac{5w\lambda}{4}} \right),\\
    & g_5(w,\tfrac{3}{5})=w+\frac{1}{5\lambda}\left(3-3e^{-\frac{5w\lambda}{3}}+4\widetilde{S}(\tfrac{5}{4}\lambda)e^{-\frac{5w\lambda}{3}}-4\widetilde{S}(\tfrac{5}{4}\lambda)e^{-\frac{5w\lambda}{4}} \right),\\
    & g_5(w,\tfrac{2}{5})=w+\frac{6}{5\lambda}\left(1-e^{-\frac{5w\lambda}{2}}\right)+\frac{8}{5\lambda}\left(\widetilde{S}(\tfrac{5}{4}\lambda)^2e^{-\frac{5w\lambda}{2}}-\widetilde{S}(\tfrac{5}{4}\lambda)^2e^{-\frac{5w\lambda}{4}} \right)+\frac{9}{5\lambda}\left(\widetilde{S}(\tfrac{5}{3}\lambda)e^{-\frac{5w\lambda}{2}}+\widetilde{S}(\tfrac{5}{3}\lambda)e^{-\frac{5w\lambda}{3}}\right),\\
    &\qquad +\frac{12}{5\lambda}\left(\widetilde{S}(\tfrac{5}{4}\lambda)\widetilde{S}(\tfrac{5}{3}\lambda)e^{-\frac{5w\lambda}{3}}-\widetilde{S}(\tfrac{5}{4}\lambda)\widetilde{S}(\tfrac{5}{3}\lambda)e^{-\frac{5w\lambda}{2}} \right)\\
    & g_5(w,\tfrac{1}{5})=w+\frac{2}{\lambda}\left(1-e^{-5w\lambda} \right)+\frac{32\widetilde{S}(\tfrac{5}{4}\lambda)^3}{15\lambda}\left(e^{-5w\lambda}-e^{-\tfrac{5w\lambda}{4}}\right)+\frac{27}{10\lambda}\left(\widetilde{S}(\tfrac{5}{3}\lambda)^2e^{-5w\lambda}+\widetilde{S}(\tfrac{5}{3}\lambda)^2e^{-\tfrac{5w\lambda}{3}} \right)\\
    &\qquad + \frac{18}{5\lambda}\left(\widetilde{S}(\tfrac{5}{3}\lambda)\widetilde{S}(\tfrac{5}{2}\lambda)e^{-\tfrac{5w\lambda}{2}}-\widetilde{S}(\tfrac{5}{3}\lambda)\widetilde{S}(\tfrac{5}{2}\lambda)e^{-5w\lambda}-\widetilde{S}(\tfrac{5}{4}\lambda)\widetilde{S}(\tfrac{5}{3}\lambda)^2e^{-5w\lambda}+\widetilde{S}(\tfrac{5}{4}\lambda)\widetilde{S}(\tfrac{5}{3}\lambda)^2e^{-\tfrac{5w\lambda}{3}} \right)\\
    &\qquad + \frac{12}{5\lambda}\left(\widetilde{S}(\tfrac{5}{2}\lambda)e^{-5w\lambda}-\widetilde{S}(\tfrac{5}{2}\lambda)e^{-\tfrac{5w\lambda}{2}} \right)+\frac{16}{5\lambda}\left(\widetilde{S}(\tfrac{5}{4}\lambda)^2\widetilde{S}(\tfrac{5}{2}\lambda)e^{-\tfrac{5w\lambda}{2}}-\widetilde{S}(\tfrac{5}{4}\lambda)^2\widetilde{S}(\tfrac{5}{2}\lambda)e^{-5w\lambda}\right)\\
    &\qquad + \frac{32}{15\lambda}\left(\widetilde{S}(\tfrac{5}{4}\lambda)^3e^{-5w\lambda}-\widetilde{S}(\tfrac{5}{4}\lambda)^3e^{-\tfrac{5w\lambda}{4}} \right)
\end{flalign*}
The affine-drift test function $h_5$ defined in \Cref{def:isqk_affine} is 
\begin{flalign*}
    & h_5(0,0)=0, &&\\
    & h_5(w,1)=w^2,\\
    & h_5(w,\tfrac{4}{5})=w^2+\dfrac{2w}{5\lambda}+\dfrac{8}{25\lambda^2}\left( 
    -1+e^{-\frac{5w\lambda}{4}}
    \right),\\
    & h_5(w,\tfrac{3}{5})=w^2+\dfrac{6w}{5\lambda}+\dfrac{2\E[S]}{5\lambda}-\dfrac{2\E[S]}{5\lambda}e^{-\frac{5w\lambda}{3}}+\dfrac{26}{25\lambda^2}\left(e^{-\frac{5w\lambda}{2}}-1 \right)+\dfrac{32\widetilde{S}(\tfrac{5}{4}\lambda)}{25\lambda^2}\left(e^{-\frac{5w\lambda}{4}}-e^{-\frac{5w\lambda}{3}} \right),\\
    & h_5(w,\tfrac{2}{5})=w^2+\dfrac{12w}{5\lambda}+\dfrac{8\E[S]}{5\lambda}\left( 
    1-e^{-\frac{5w\lambda}{2}}
    \right)+\dfrac{6\widetilde{S}(\tfrac{5}{3}\lambda)\E[S]}{5\lambda}\left(e^{-\frac{5w\lambda}{2}}-e^{-\frac{5w\lambda}{3}} \right)+\dfrac{2}{\lambda^2}\left(e^{-\frac{5w\lambda}{2}}-1\right)\\
    &\qquad+\dfrac{64\widetilde{S}(\tfrac{5}{4}\lambda)^2}{25\lambda^2}\left(e^{-\frac{5w\lambda}{4}}-e^{-\frac{5w\lambda}{2}} \right)+\dfrac{78\widetilde{S}(\tfrac{5}{3}\lambda)}{25\lambda^2}\left(e^{-\frac{5w\lambda}{3}}-e^{-\frac{5w\lambda}{2}}\right)+\dfrac{96\widetilde{S}(\tfrac{5}{4}\lambda)\widetilde{S}(\tfrac{5}{3}\lambda)}{25\lambda^2}\left(e^{-\frac{5w\lambda}{2}}-e^{-\frac{5w\lambda}{3}} \right),\\
    & h_5(w,\tfrac{1}{5})=w^2+\dfrac{4w}{\lambda}+\dfrac{4\E[S]}{\lambda}-\dfrac{4\E[S]e^{-5w\lambda}}{\lambda}+\dfrac{9\widetilde{S}(\tfrac{5}{3}\lambda)^2\E[S]}{\lambda}\left(e^{-5w\lambda}-e^{-\frac{5w\lambda}{3}} \right)\\
    &\qquad+\dfrac{12\widetilde{S}(\tfrac{5}{3}\lambda)\widetilde{S}(\tfrac{5}{2}\lambda)\E[S]}{5\lambda}\left(e^{-\frac{5w\lambda}{2}}-e^{-5w\lambda} \right)+\dfrac{16\widetilde{S}(\tfrac{5}{2}\lambda)\E[S]}{5\lambda}\left(e^{-5w\lambda}-e^{-\frac{5w\lambda}{2}} \right)\\
    &\qquad+\dfrac{14}{5\lambda^2}\left(e^{-5w\lambda}-1\right)+\dfrac{256\widetilde{S}(\tfrac{5}{4}\lambda)^3}{75\lambda^2}\left(e^{-\frac{5w\lambda}{4}}-e^{-5w\lambda} \right)+\dfrac{117\widetilde{S}(\tfrac{5}{3}\lambda)^2}{25\lambda^2}\left( e^{-\frac{5w\lambda}{2}}-e^{-5w\lambda}\right)\\
    &\qquad+\dfrac{144\widetilde{S}(\tfrac{5}{4}\lambda)\widetilde{S}(\tfrac{5}{3}\lambda)^2}{25\lambda^2}\left(e^{-5w\lambda}-e^{-\frac{5w\lambda}{3}} \right)+\dfrac{4\widetilde{S}(\tfrac{5}{2}\lambda)}{\lambda^2}\left(e^{-\frac{5w\lambda}{3}}-e^{-5w\lambda} \right)+\\
    &\qquad+\dfrac{128\widetilde{S}(\tfrac{5}{4}\lambda)^2\widetilde{S}(\tfrac{5}{2}\lambda)}{25\lambda^2}\left(e^{-5w\lambda}-e^{-\frac{5w\lambda}{2}} \right)\dfrac{156\widetilde{S}(\tfrac{5}{3}\lambda)\widetilde{S}(\tfrac{5}{2}\lambda)}{25\lambda^2}\left( e^{-5w\lambda}-e^{-\frac{5w\lambda}{2}}\right)\\
    &\qquad+\dfrac{192\widetilde{S}(\tfrac{5}{4}\lambda)\widetilde{S}(\tfrac{5}{3}\lambda)\widetilde{S}(\tfrac{5}{2}\lambda)}{25\lambda^2}\left(e^{-\frac{5w\lambda}{2}}-e^{-5w\lambda} \right).
\end{flalign*}
\section{SEK Policy}
\label{apdx:sek}

To compare the Uncertainty Improvement Ratio (UIR) of the SEK policy with our ISQ-Recycling lower bound relative to the MixEx bound, we evaluate the Practical-SEK policy (see Definition 7.1 of \citet{Grosof2026}). 
We consider the setting with $k=2$ and exponentially distributed $Exp(1)$ job sizes. 
To obtain an improved upper bound relative to SRPT-2, we run the Practical-SEK policy using multiple thresholds, $\epsilon \in \{0.1, 0.2, 0.5, 1.0\}$, and take, for each load $\rho$, the minimum of the four simulated mean response times. 

Given that the simulated Practical-SEK response time is a novel upper bound, the definition of UIR (\Cref{def:uir}) is given by,
\begin{align*}
    \mathrm{UIR} = \frac{\E[T_{\text{upper}}] - \E[T_{\text{novel-upper}}]}{\E[T_{\text{upper}}] - \E[T_{\text{lower}}]}.
\end{align*}

As shown in \Cref{fig:sek_uir}, the ISQ-Recycling bound achieves substantially higher UIR values across $\rho > 0.35$, attaining a maximum UIR of 34\%, while the Practical-SEK policy achieves only 3\% maximum UIR. 
Overall, ISQ-Recycling provides roughly an order-of-magnitude improvement in the maximum UIR compared to Practical-SEK.

\end{document}